\pdfoutput=1

\documentclass[10pt]{article} 

\usepackage{algorithmic}
\usepackage{algorithm}

\usepackage{graphicx}
\usepackage{amsmath}
\usepackage{amsthm}
\usepackage{amsbsy}
\usepackage{amssymb}
\usepackage{amsfonts}
\usepackage[dvips]{lscape}

\theoremstyle{plain}
\newtheorem{theorem}{Theorem}[section]

\newtheorem{lemma}[theorem]{Lemma}
\newtheorem{corollary}[theorem]{Corollary}
\newtheorem{remark}[theorem]{Remark}
\newtheorem{definition}[theorem]{Definition}

\DeclareMathOperator*{\argmin}{argmin}
\DeclareMathOperator*{\rank}{rank}
\DeclareMathOperator*{\nullspace}{nullspace}

\DeclareMathOperator*{\card}{card}
\DeclareMathOperator*{\size}{size}

\DeclareMathOperator*{\low}{low}
\DeclareMathOperator*{\dist}{dist}
\DeclareMathOperator*{\rad}{rad}
\DeclareMathOperator*{\radZ}{radZ}
\DeclareMathOperator*{\diam}{diam}
\DeclareMathOperator*{\img}{img}

\DeclareMathOperator*{\length}{length}
\DeclareMathOperator*{\vertex}{vert}

\newcommand\inlineM[1]{\left[ \begin{smallmatrix} #1 \end{smallmatrix} \right]}

\title{Measuring and Localizing Homology Classes}
\author{
Daniel Freedman\\
Rensselaer Polytechnic Institute\\
freedman@cs.rpi.edu\\
\and
Chao Chen\\
Rensselaer Polytechnic Institute\\
chenc3@cs.rpi.edu\\
}

\begin{document}

\maketitle

%%%%%%%%%%%%%%%%%%%%%%%%%%%%%%%%%%%%%%%%%%%%%%%%%%
%%%%%%%%%%%%%%%%%%%%%%%%%%%%%%%%%%%%%%%%%%%%%%%%%%

\begin{abstract}
We develop a method for measuring and localizing homology classes.  This involves two problems.  First, we define relevant notions of size for both a homology class and a homology group basis, using ideas from relative homology.  Second, we propose an algorithm to compute the optimal homology basis, using techniques from persistent homology and finite field algebra.  Classes of the computed optimal basis are localized with cycles conveying their sizes. The algorithm runs in $O(\beta^4 n^3 \log^2 n)$ time, where $n$ is the size of the simplicial complex and $\beta$ is the Betti number of the homology group.  
\end{abstract}

\section{Introduction}
\label{sec:intro}

In recent years, the problem of computing the topological features of a space has drawn much attention.
There are two reasons for this.  The first is a general observation: compared with geometric features, topological features are more qualitative and global, and tend to be more robust.  If the goal is to characterize a space, therefore, features which incorporate topology seem to be good candidates.  

The second reason is that topology plays an important role in a number of applications.  Researchers in graphics need topological information to facilitate parameterization of surfaces and texture mapping \cite{EricksonH04,CarnerJGQ05}. In the field of sensor networks, the use of homological tools is crucial for certain coverage problems \cite{deSilvaG06}.  Computational biologists use topology to study protein docking and folding problems \cite{AgarwalEHW06,Cohen-SteinerEM06}. Finally, topological features are especially important in high dimensional data analysis, where purely geometric tools are often deficient, and full-blown space reconstruction is expensive and often ill-posed \cite{Carlsson05,Ghrist}.

Once we are able to compute topological features, a natural problem is to rank the features according to their importance.  The significance of this problem can be justified from two perspectives.  First, unavoidable errors are introduced in data acquisition, in the form of traditional signal noise, and finite sampling of continuous spaces.  These errors may lead to the presence of many small topological features that are not ``real'', but are simply artifacts of noise or of sampling \cite{WoodHDS04}.  Second, many problems are naturally hierarchical.  This hierarchy -- which is a kind of multiscale or multi-resolution decomposition -- implies that we want to capture the large scale features first.  See Figure \ref{fig:2torusBigSmall} for examples.

\begin{figure}[hbtp]
    \centerline{
    \begin{tabular}{cc}
		\includegraphics[width=0.25\textwidth]{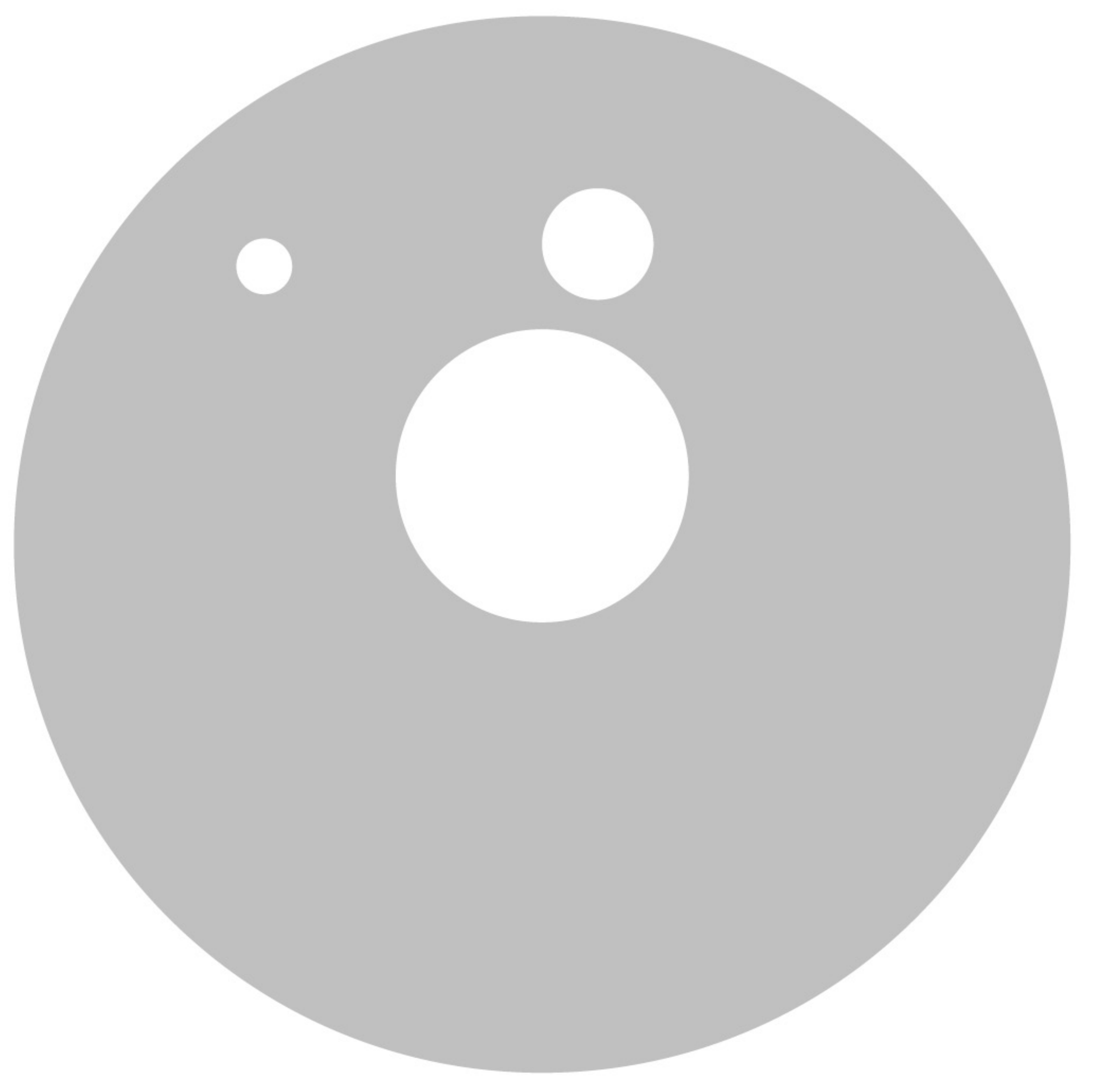} &
		\includegraphics[width=0.4\textwidth]{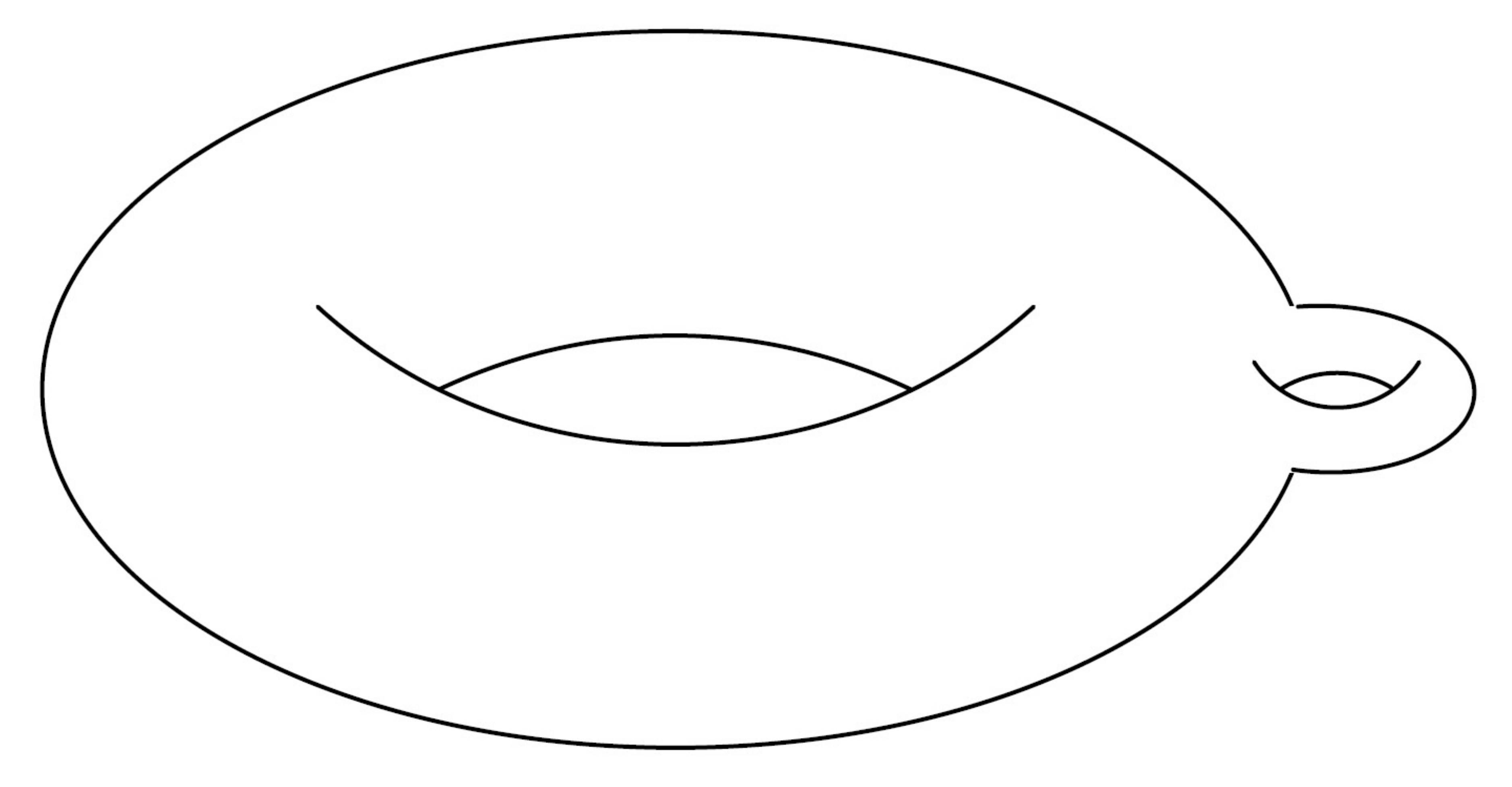}
    \end{tabular}}
    \caption{ A disk with three holes and a 2-handled torus are really more
    like an annulus and a 1-handled torus, respectively, because the large features 
    are more important.}
    \label{fig:2torusBigSmall}
\end{figure}

There are a variety of ways of characterizing topological spaces in the literature, including fundamental groups, homology groups, and the Euler characteristic. In this paper, we concentrate on homology groups as they are relatively straightforward to compute in general dimension, and provide a decent amount of information (more, say, than a coarse measure like the Euler characteristic).

Ranking the homology classes according to their importance involves the following three subproblems. 

\begin{enumerate}

\item {\bf Measuring the size of a homology class:} We need a way to quantify the size of a given homology class, and this size measure should agree with intuition.  For example, in Figure \ref{fig:badGenerator} (center), the measure should be able to distinguish the one large class (of the 1-dimensional homology group) from the two smaller classes.  Furthermore, the measure should be easy to compute, and applicable to homology groups of any dimension.

\item {\bf Localizing a homology class:} Given the size measure for a homology class, we would like to find a representative cycle from this class which, in a precise sense, has this size.  For example, in Figure \ref{fig:badGenerator} (center), the cycles $z_1$ and $z_2$ are well-localized representatives of their respective homology classes; whereas $z_3$ is not.

\begin{figure}[hbtp]
    \centerline{
    \begin{tabular}{ccc}
		\includegraphics[width=0.28\textwidth]{holeDisk0.pdf} &
		\includegraphics[width=0.28\textwidth]{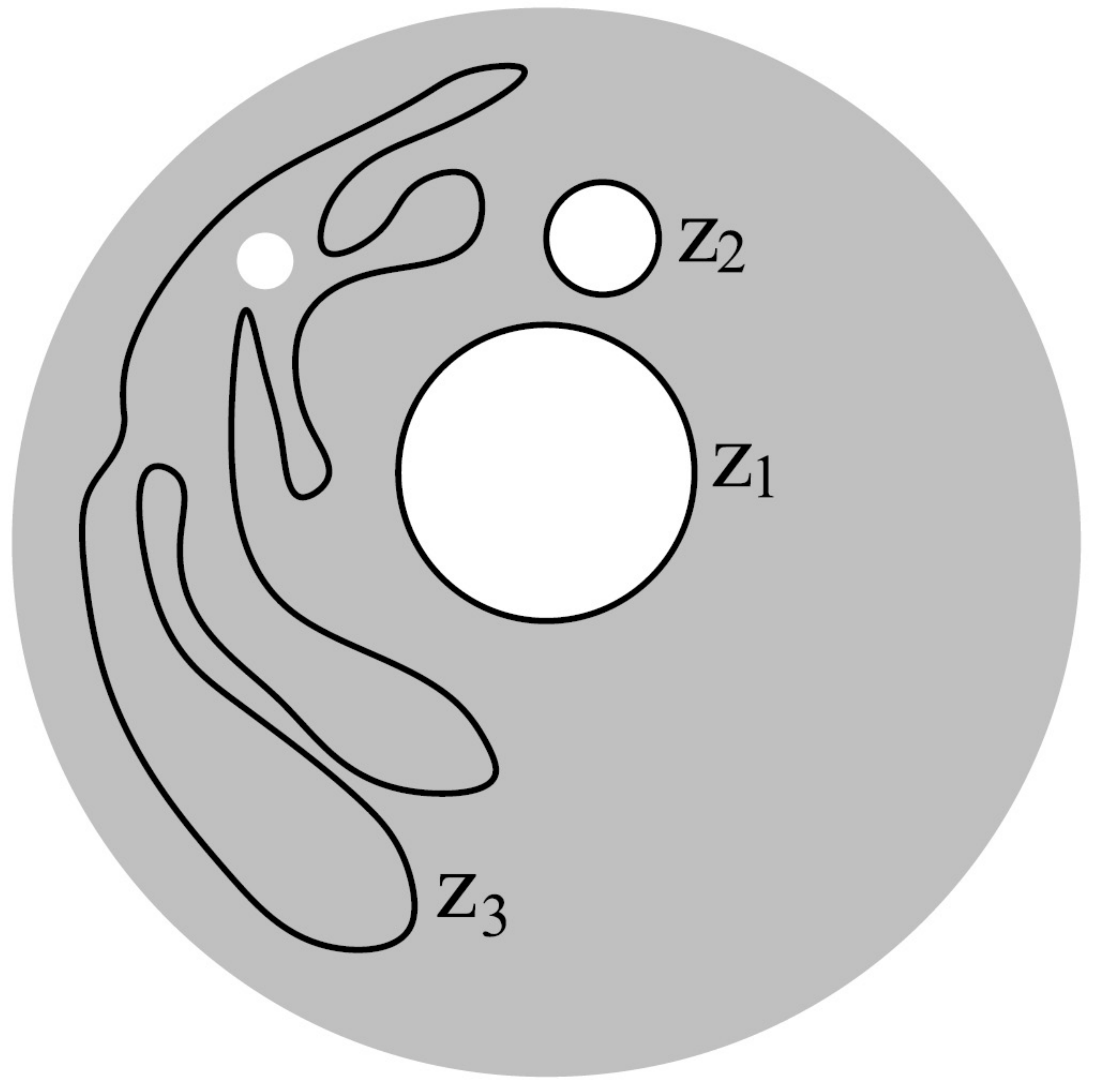} &
		\includegraphics[width=0.28\textwidth]{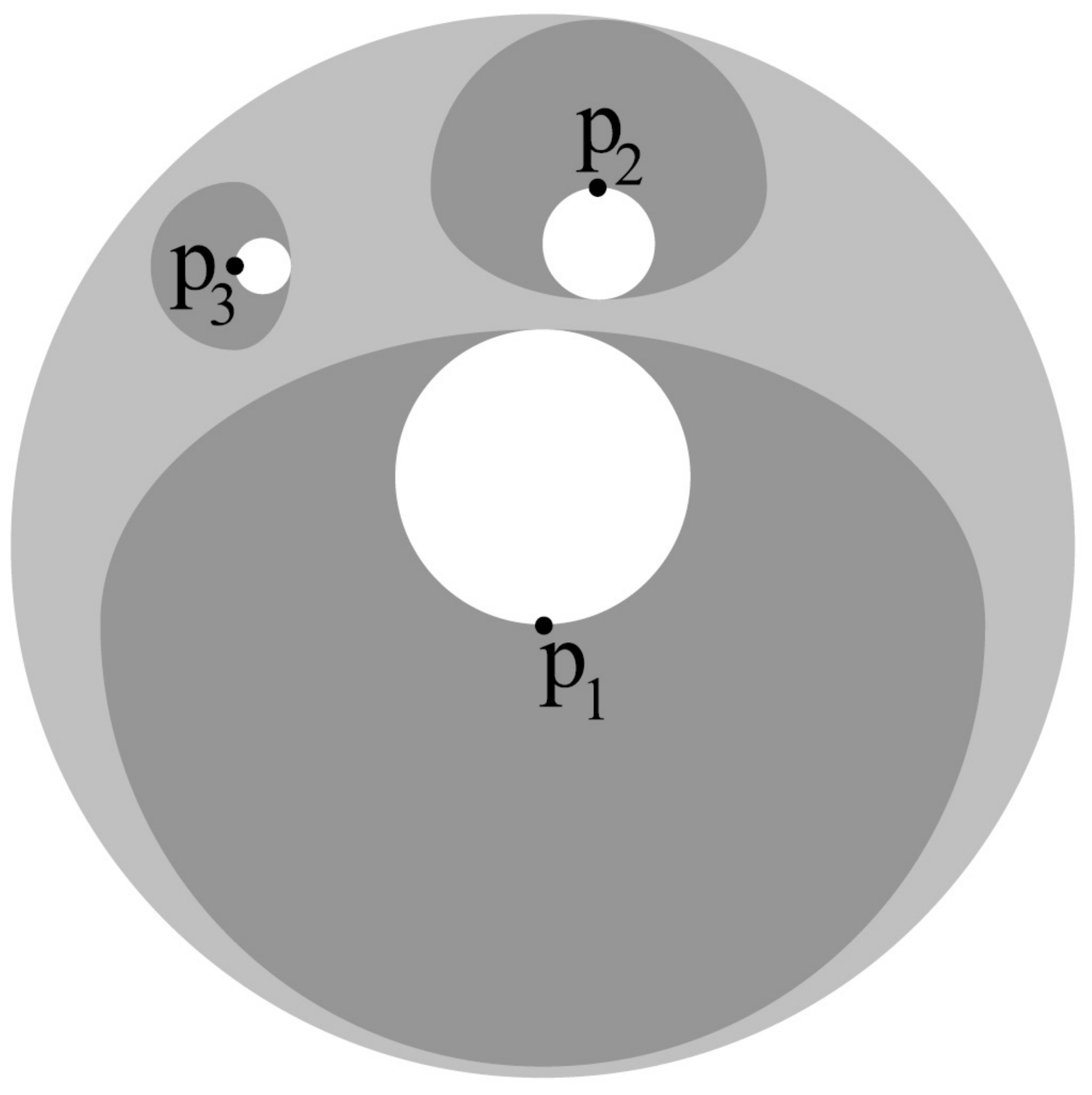}
    \end{tabular}}
    \caption{ A disk with three holes. Left: the underlying topological space. 
    Center: cycles $z_1$ and $z_2$ 
    convey the size of their respective homology classes; $z_3$ does not. 
    Right: geodesic balls measuring the 1-dimensional homology classes (used in
    Section \ref{sec:sizeDef}).}
    \label{fig:badGenerator}
\end{figure}

\item {\bf Choosing a basis for a homology group:}  We would like to choose a ``good'' set of homology classes to be the generators for the homology group (of a fixed dimension).  Suppose that $\beta$ is the dimension of this group, and that we are using $\mathbb{Z}_2$ coefficients; then there are $2^\beta-1$ nontrivial homology classes in total. For a basis, we need to choose a subset of $\beta$ of these classes, subject to the constraint that these $\beta$ generate the group.  The criterion of goodness for a basis is based on an overall size measure for the basis, which relies in turn on the size measure for its constituent classes.  For instance, in Figure \ref{fig:3circles}, we must choose three from the seven nontrivial $1$-dimensional homology classes: $\{[z_1], [z_2], [z_3], [z_1]+[z_2], [z_1]+[z_3], [z_2]+[z_3], [z_1]+[z_2]+[z_3]\}$.  In this case, the intuitive choice is $\{ [z_1], [z_2], [z_3] \}$, as this choice reflects the fact that there is really only one large cycle.

\begin{figure}[hbtp]
    \centering
    \includegraphics[width=0.45\textwidth]{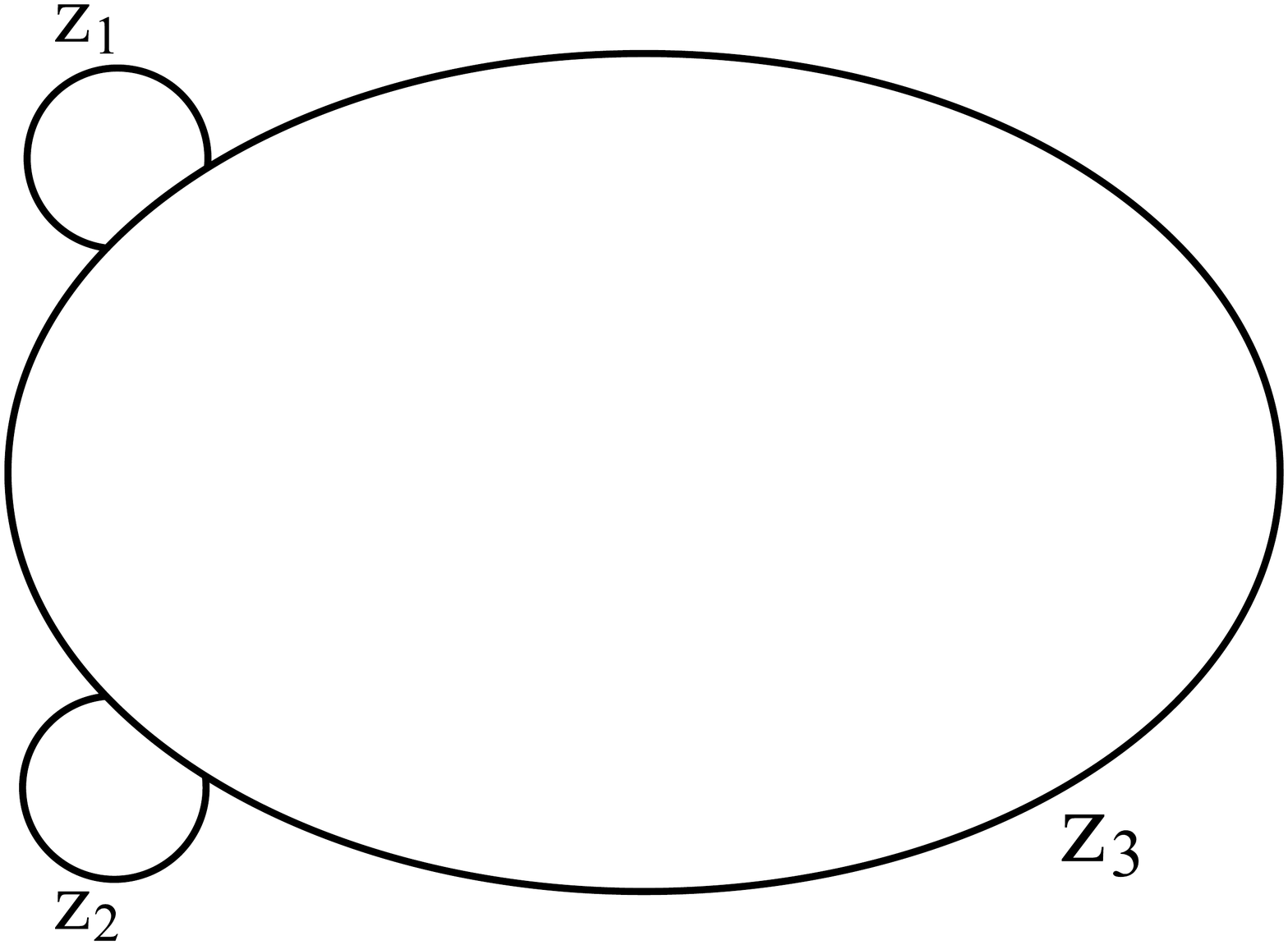}
    \caption{ A topological space formed from three circles.  See accompanying discussion in the text.}
    \label{fig:3circles}
\end{figure}

\end{enumerate}

%%%%%%%%%%%%%%%%%%%%%%%%%%%%%%%%%%%%%%%%%%%%%%%%%%%
\subsection{Related Works}
There is much work that has been done in the general field of computational topology \cite{cs.CG/9909001}.  Examples include fast algorithms for computing Betti numbers \cite{DelfinadoE93,Friedman96}, as well as techniques for relating topological spaces to their approximations \cite{NiyogiSW06,ChazalL05}; where the latter usually derive from sampled versions of the spaces.  However, in the following we will focus only on the areas of computational topology which are most germane to the current study: persistent homology and algorithms for localizing topological features.  Note that a more formal review of persistence will be given in Section \ref{sec:persistence}.

\paragraph{Persistent Homology}
Persistent homology \cite{EdelsbrunnerLZ02,Cohen-SteinerEH07,ZomorodianC05,CarlssonZ07} is designed to track the persistences of homological features over the course of a filtration of a topological space.  At first blush, it might seem that the powerful techniques of this theory are ideally suited to solving the problems we have set out.  However, due to their somewhat different motivation, these techniques do not quite yield a solution.  There are two reasons for this.  First, the persistence of a feature depends not only on the space in which the feature lives, but also on the filtering function chosen.  In the absence of a geometrically meaningful filter, it is not clear whether the persistence of a feature is a meaningful representation of its size.  Second, and more importantly, the persistence only gives information for homology classes which ultimately die; for classes which are intrinsically part of the topological space, and which thus never die, the persistence is infinite.  However, it is precisely these \emph{essential} (or non-persistent) classes that we care about.  

In more recent work, Cohen-Steiner et al.~\cite{Cohen-SteinerEH} have extended persistent homology in such a way that essential homology classes also have finite persistences. 
This extension serves to complete the theory and has some nice properties like stability, duality and symmetry for triangulated manifolds. However, the persistences thus computed still depend on the filter function, and furthermore, do not always seem to agree with an intuitive notion of size. See Figure \ref{fig:extPersist}.
\begin{figure}[hbtp]
    \centering
		\includegraphics[width=0.4\textwidth]{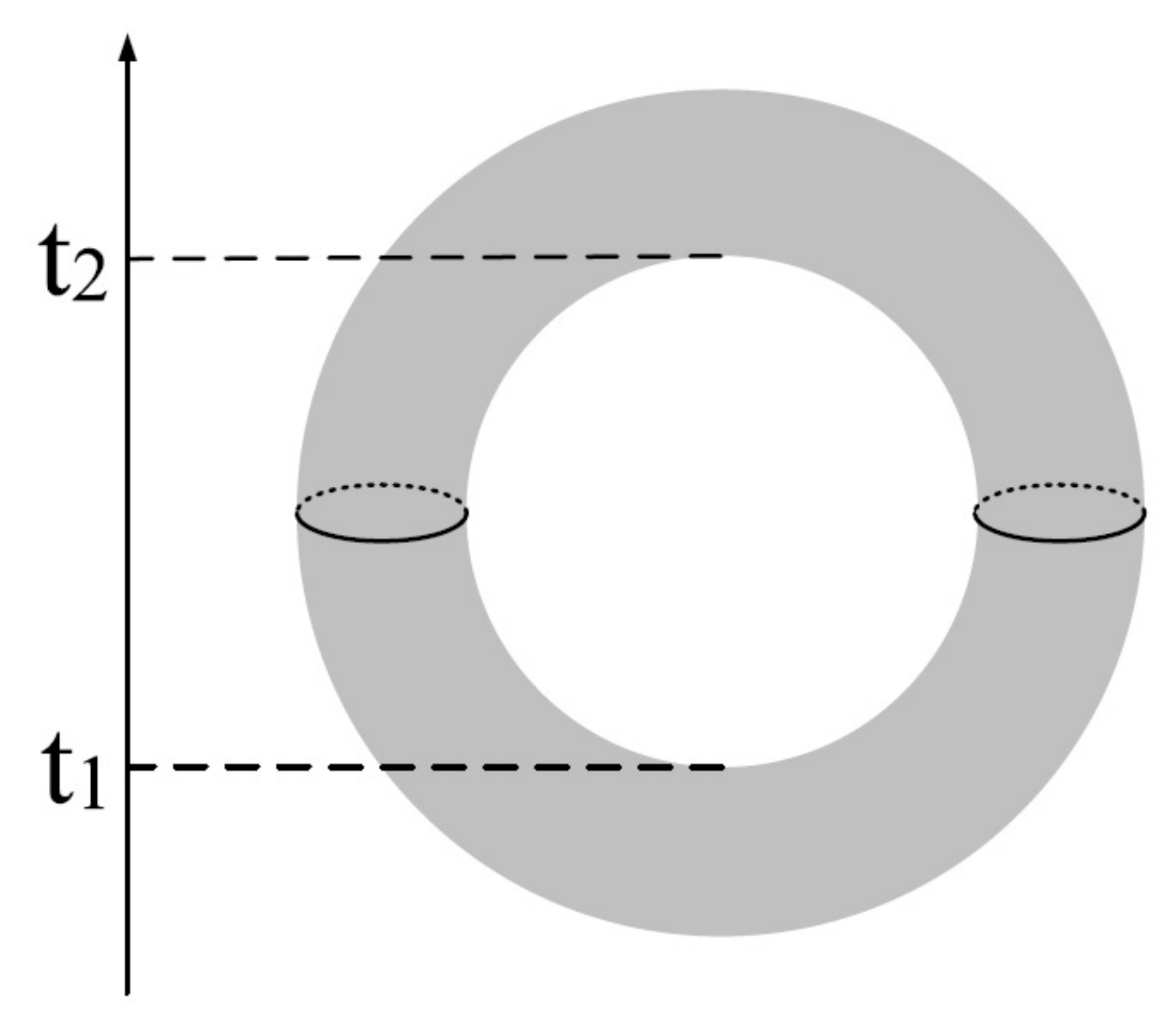}
    \caption{ Computing the extended persistent homology of a torus using the height function as the filter function. The (birth,death time) pairs of the two 1-dimensional homology classes are $(t_1,t_2)$ and $(t_2,t_1)$, respectively. The persistences are not consistent with our intuition of their sizes.}
    \label{fig:extPersist}
\end{figure}

\paragraph{Localization of Topological Features}
Zomorodian and Carlsson \cite{ZomorodianC07} take a different approach to solving the localization problem.  Their method starts with a topological space and a cover, a set of spaces whose union contains the original space. A blowup complex is built up which contains homology classes of all the spaces in the cover. The authors then use persistent homology to
identify homology classes in the blowup complex which correspond to a same homology class in the given topological space. The persistent homology algorithm produces a complete set of generators for the relevant homology group, which forms a basis for the group. However, both the quality of the generators and the complexity of the algorithm depend strongly on the choice of cover; there is, as yet, no suggestion of a canonical cover.

Using Dijkstra's shortest path algorithm, Erickson and Whittlesey \cite{EricksonW05} showed how to localize a one-dimensional homology class with its shortest cycle. Although not explicitly mentioned,
the length of this shortest cycle can be deemed as a measure of the size of its homology class. They proved, by an application of matroid theory, that finding $\beta$ linearly independent homology classes whose sizes have the smallest sum can be achieved by a greedy method, namely, finding the smallest  homology classes one by one,  subject to a linear independence constraint. Their algorithm takes $O(n^2\log n + n^2 \beta + n \beta^3)$ or $O(n^2\beta + n\beta^3)$ if $\beta$ is nearly linear in $n$.  The authors also show how the idea carries over to finding the optimal generators of the first fundamental group, though the proof is considerably harder in this case.  Note that this work is restricted to $1$-dimensional homology classes in a $2$-dimensional topological space.  A similar measure was used by Wood et al.~\cite{WoodHDS04} to remove topological noise of $2$-dimensional surface. This work also suffers from the dimension restriction.

%%%%%%%%%%%%%%%%%%%%%%%%%%%%%%%%%%%%%%%%%%%%%%%%%%%

\subsection{Our Contributions}

In this paper, we solve the three problems listed in Section \ref{sec:intro}, namely, 
measuring the size of homology classes, localizing classes, and choosing
a basis for a homology group. We define a size measure
for homology classes, based on relative homology, using geodesic distance.
This solves the first problem.
For the second problem, we localize homology classes with cycles
which are strongly related to the size measure just defined.
We solve the third problem by choosing the set of linearly independent homology 
classes whose sizes have the minimal sum. 
The time complexity of our algorithm is $O(\beta^4 n^3\log^2 n)$, where
$n$ is the cardinality of the given simplicial complex, and $\beta$ is 
the dimension of the homology group. We assume the input of our
algorithm is a simplicial complex $K$, i.e.~a triangulation of 
the given topological space.

\paragraph{ Size measure and localization.} In section 3, we define
the {\it size} of a homology class $h$, $S(h)$, as the radius of the 
smallest geodesic ball within the topological space which carries 
a cycle of $h$, $z_0\in h$. Here a {\it geodesic ball}, $B_p^r$,
is the subset of the topological space consisting of points whose geodesic
distance from the point $p$ is no greater than $r$.
The intuition behind this definition will be further elaborated in Section \ref{sec:sizeDef}.
Any cycle of $h$ lying within this smallest geodesic ball is 
a {\it localized cycle} of $h$.

\paragraph{ Optimal homology basis.}
Although there are $2^{\beta}-1$ nontrivial homology classes, only $\beta$ of them are 
needed to construct the homology group, subject to the constraint that 
these classes generate the group. 
We choose to compute the set whose sizes have the minimal sum, which we call the 
{\it optimal homology basis}. This basis contains as few large homology
classes as possible, and thus captures important features effectively.

\paragraph{ Computing the smallest class.}
To compute the smallest nontrivial homology class, we find the smallest geodesic ball, $B_{min}$, 
which carries any nonbounding cycle of the given simplicial complex $K$. 
To find $B_{min}$, we visit all of the vertices of $K$ in turn.
For each vertex $p$, we compute the persistent homology using the geodesic distance
from $p$ as a filter. This yields the smallest geodesic ball centered on $p$
carrying any nonbounding cycle of $K$, namely, $B_p^{r(p)}$. The ball with the
smallest $r(p)$ is exactly $B_{min}$.  Once we find $B_{min}$,
its radius, $r_{min}$, is the size of the smallest class. Any nonbounding cycle of $K$
carried by $B_{min}$ is a localized cycle of this class,
and can be computed by a reduction-style algorithm. 

\paragraph{ Computing the optimal homology basis.}
We use matroid theory to prove that the optimal homology basis can be computed
by a greedy method.  We first compute the smallest homology class of the given simplicial 
complex $K$, as described above. We then destroy this class
by sealing up one of its cycles with new simplices. Next, we compute the smallest 
homology class of the updated simplicial complex, $K'$, which is the second smallest
class of the optimal homology basis of $K$. We then destroy this class and proceed
to compute the third smallest class. The whole basis is computed in $\beta$ rounds.
Theorem \ref{thm:sealingUp} establishes that this sealing technique yields the optimal homology basis.
The time to compute the optimal homology basis is $O(\beta^4 n^4)$.

\paragraph{ An improvement using finite field linear algebra.} In computing the smallest geodesic ball 
$B_{min}$, we may avoid explicit
computation of $B_{p}^{r(p)}$ for every $p$. 
Instead, Theorem \ref{thm:neighborClose} suggests we visit all of the vertices in a breadth-first 
fashion. For the root of the breadth-first tree, we use the explicit algorithm; for the rest of the vertices, we need only check whether
a specific geodesic ball carries any nonbounding cycle of $K$.
This latter task is not straightforward, as some of the nonbounding cycles in this ball may be
boundaries in $K$. We use Theorem 
\ref{thm:rank} to reduce this problem to rank computations of sparse matrices 
over the $\mathbb{Z}_2$ field. 
The time to compute the optimal homology basis with this improvement is $O(\beta^4 n^3\log^2 n)$.

\paragraph{ Consistency with existing results.}
We prove in Section \ref{sec:lowDimResult} that our result is consistent with the low dimensional 
optimal result of Erickson and Whittlesey \cite{EricksonW05}.

%%%%%%%%%%%%%%%%%%%%%%%%%%%%%%%%%%%%%%%%%%%%%%%%%%%
%%%%%%%%%%%%%%%%%%%%%%%%%%%%%%%%%%%%%%%%%%%%%%%%%%%

\section{Preliminaries}
In this section, we briefly describe the background necessary for our work,
including a discussion of simplicial complexes, homology groups, 
persistent homology, and relative homology. Please
refer to \cite{Munkres84} for further details in algebraic topology,
and \cite{EdelsbrunnerLZ02,ZomorodianC05,Cohen-SteinerEH07,CarlssonZ07} for persistent
homology. For simplicity, we restrict our discussion to the combinatorial
framework of simplicial homology in the 
$\mathbb{Z}_2$ field.

%%%%%%%%%%%%%%%%%%%%%%%%%%%%%%%%%%%%%%%%%%%%%%%%%%%

\subsection{Simplicial Complex}

A $d${\it -dimensional simplex} or $d$-{\it simplex}, $\sigma$, 
is the convex hull of $d+1$ {\it affinely independent vertices},
which means for any of these vertices, $v_i$, the $d$ vectors
$v_j-v_i$, $j\neq i$, are linearly independent.
A $0$-simplex, $1$-simplex, $2$-simplex and $3$-simplex are
a vertex, edge, triangle and tetrahedron, respectively.
The convex hull of a nonempty subset of vertices of $\sigma$ is its {\it face}. 
A {\it simplicial complex} $K$ is a finite set of simplices that satisfies 
the following two conditions.
\begin{enumerate}
\item Any face of a simplex in $K$ is also in $K$.
\item The intersection of any two simplices in $K$ is either empty or is a 
face for both of them. 
\end{enumerate}
The {\it dimension} of a simplicial complex is the highest dimension of
its simplices. If a subset $K_0\subseteq K$ is a simplicial complex,
it is a {\it subcomplex} of $K$.

%%%%%%%%%%%%%%%%%%%%%%%%%%%%%%%%%%%%%%%%%%%%%%%%%%%

\subsection{Homology Groups}  

Within a given simplicial complex $K$, a $d${\it -chain} is a formal sum
$d$-simplices in $K$, $c=\sum_{\sigma\in K}a_{\sigma}\sigma$, $a_\sigma \in \mathbb{Z}_2$. All the
$d$-chains form the {\it group of $d$-chains}, $\mathsf{C}_d(K)$.
The {\it boundary} of a
$d$-chain is the sum of the $(d-1)$-faces of all the $d$-simplices in the
chain. The boundary operator $\partial_d:\mathsf{C}_d(K)\rightarrow\mathsf{C}_{d-1}(K)$ 
is a group homomorphism.

A $d${\it -cycle} is a $d$-chain without boundary. The set of $d$-cycles forms
a subgroup of the chain group, which is the kernel of the boundary operator, $\mathsf{Z}_d(K) = \ker(\partial_d)$.
A $d${\it -boundary} is the boundary of a $(d+1)$-chain. The set of $d$-boundaries
forms a group, which is the image of the boundary operator, $\mathsf{B}_d(K)=\img(\partial_{d+1})$.
It is not hard to see that a $d$-boundary is also a $d$-cycle. Therefore,
$\mathsf{B}_d(K)$ is a subgroup of $\mathsf{Z}_d(K)$. 
A $d$-cycle which is not a $d$-boundary, $z\in \mathsf{Z}_d(K)\backslash\mathsf{B}_d(K)$, 
is a {\it nonbounding cycle}. 

The $d${\it -dimensional homology group}
is defined as the quotient group $\mathsf{H}_d(K)=\mathsf{Z}_d(K)/\mathsf{B}_d(K)$.
An element in $\mathsf{H}_d(K)$ is a {\it homology class}, which is a coset of $\mathsf{B}_d(K)$,
$[z]=z+\mathsf{B}_d(K)$ for some $d$-cycle $z\in \mathsf{Z}_d(K)$.
If $z$ is a $d$-boundary, $[z]=\mathsf{B}_d(K)$ is the identity element of $\mathsf{H}_d(K)$.
Otherwise, when $z$ is a nonbounding cycle, $[z]$ is a {\it nontrivial homology class}
and $z$ is called a {\it representative cycle} of $[z]$.
Cycles in the same homology class are {\it homologous} to each other, which means
their difference is a boundary.

The dimension of the homology group, which is referred to as the
{\it Betti number}; $\beta_d=\dim(\mathsf{H}_d(K))=\dim(\mathsf{Z}_d(K))-\dim(\mathsf{B}_d(K))$.
It can be computed with a reduction algorithm based
on row and column operations of the boundary matrices \cite{Munkres84}.
Various reduction algorithms have been devised for different purposes \cite{KaczynskiMS98,EdelsbrunnerLZ02,ZomorodianC05}.

The following notation will prove convenient.
We say that a $d$-chain $c\in \mathsf{C}_d(K)$ is {\it carried by} 
a subcomplex $K_0$ when all the $d$-simplices of $c$ belong to $K_0$,
formally, $c\subseteq K_0$.
We denote $\vertex(K)$ as the set
of vertices of the simplicial complex $K$, $\vertex(c)$ as that of the chain $c$.

In this paper, we focus on the simplicial homology over the finite field $\mathbb{Z}_2$.
In this case, a chain corresponds to a $n_d$-dimensional vector, where
$n_d$ is the number of $d$-simplices in $K$. Computing the boundary of
a $d$-chain corresponds to multiplying the chain vector with a boundary matrix
$[b_1,...,b_{n_d}]$, whose column 
vectors are boundaries of $d$-simplices in $K$. By slightly abusing the notation,
we call the boundary matrix $\partial_d$.

%%%%%%%%%%%%%%%%%%%%%%%%%%%%%%%%%%%%%%%%%%%%%%%%%%%

\subsection{Persistent Homology}  
\label{sec:persistence}

Given a topological space $\mathbb{X}$ and a {\it filter function}
$f:\mathbb{X}\rightarrow \mathbb{R}$, {\it persistent homology} studies the homology
classes of the sublevel sets, $\mathbb{X}^t=f^{-1}(-\infty,t]$. 
A nontrivial homology class in $\mathbb{X}^{t_1}$ may become trivial
in $\mathbb{X}^{t_2}$, $t_1<t_2$, (formally, when induced by the inclusion homomorphism).
Persistent homology tries to capture this phenomenon by measuring the
times at which a homology class is born and dies. The persistence, or life time
of the class is the difference between its death and birth times.
Those with longer lives tell us something about the global structure
of the space $\mathbb{X}$, as described by the filter function. 
Note that the essential, that is, nontrivial homology 
classes of the given topological space $\mathbb{X}$ will never die. 

Edelsbrunner et al.~\cite{EdelsbrunnerLZ02} 
devised an $O(n^3)$ algorithm to compute the 
persistent homology. Its input are a simplicial complex $K$ and a filter function
$f$, which assigns each simplex in $K$ a real value.
Simplices of $K$ are sorted in ascending order according to their filter function
values. This order is actually the order in which simplices
enter the sublevel set $f^{-1}(-\infty,t]$ while $t$ increases.
For simplicity, in this paper we call this ordering the 
{\it simplex-ordering} of $K$ with regard to $f$. The output of the algorithm 
is the birth and death times of homology classes.

The algorithm performs column operations on an overall incidence matrix, 
$D$, whose rows and columns correspond to simplices in $K$. 
An entry $D(i,j)=1$ if and only if the simplex $\sigma_i$ belongs
to the boundary of the simplex $\sigma_j$. To some extent, $D$
is a big boundary matrix which can accommodate chains of arbitrary dimension.
Columns and rows of $D$ are sorted in ascending order
according to the function values of simplices. The algorithm performs the 
column reduction from left to right, recording $\low(i)$ as the lowest
nonzero entry of each column $i$. If column $i$ is reduced
to a zero column, $\low(i)$ does not exist. To reduce column $i$, we 
repeatedly find column $j$ satisfying $j<i$ and $\low(j)=\low(i)$; we then 
add column $j$ to column $i$, until column $i$ becomes a zero column
or we cannot find a qualified $j$ anymore.

The reduction of $D$ can be written as a matrix multiplication,
\begin{eqnarray}
R=DV,
\label{eqn:reduceD}
\end{eqnarray}
where $R$ is the reduced matrix and $V$ is an upper triangular matrix.
Columns of $V$ corresponding to zero columns of $R$ whose corresponding
simplices are $d$-dimensional form a basis of
the cycle group $\mathsf{Z}_d(K)$.

After the reduction, each paring, $\low(i)=j$, corresponds to a homology
class whose birth time is $f(\sigma_i)$ and death time is $f(\sigma_j)$.
A simplex $\sigma_i$ that is not paired, namely, neither $\low(i)=j$ nor
$\low(j)=i$ for any $j$, corresponds to an {\it essential homology class}, namely,
a nontrivial homology class of $K$. An essential homology class only has a birth
time, namely, $f(\sigma_i)$, and it never dies. Therefore, all the nontrivial
homology classes of $K$ have infinite persistences.

%%%%%%%%%%%%%%%%%%%%%%%%%%%%%%%%%%%%%%%%%%%%%%%%%%%

\subsection{Relative Homology} 

Given a simplicial complex $K$ and a subcomplex $K_0\subseteq K$, we may
wish to study the structure of $K$ by ignoring all the chains in $K_0$. 
We consider two $d$-chains, $c_1$ and $c_2$ to be the same if their 
difference is carried by $K_0$.
The objects we are interested in are then defined as these equivalence
classes, which form a quotient group, $\mathsf{C}_d(K,K_0)=\mathsf{C}_d(K)/\mathsf{C}_d(K_0)$.
We call it the {\it group of relative chains}, whose elements (cosets),
are called {\it relative chains}.

The boundary operator $\partial_d:\mathsf{C}_d(K)\rightarrow \mathsf{C}_{d-1}(K)$ 
induces a {\it relative boundary operator}, $\partial_d^{K_0}:\mathsf{C}_d(K,K_0)\rightarrow \mathsf{C}_{d-1}(K,K_0)$.
Analogous to
the way we define $\mathsf{Z}_d(K)$, $\mathsf{B}_d(K)$ and $\mathsf{H}_d(K)$ in $\mathsf{C}_d(K)$,
we define the {\it group of relative cycles}, the {\it group of relative boundaries} and
the {\it relative homology group} in $\mathsf{C}_d(K,K_0)$, denoted as $\mathsf{Z}_d(K,K_0)$,
$\mathsf{B}_d(K,K_0)$ and $\mathsf{H}_d(K,K_0)$, respectively.
An element in $\mathsf{Z}_d(K,K_0)\backslash\mathsf{B}_d(K,K_0)$ is a {\it nonbounding relative cycle}.

The following notation will prove convenient. We define a homomorphism 
$\phi_{K_0}:\mathsf{C}_d(K)\rightarrow \mathsf{C}_d(K,K_0)$ 
mapping $d$-chains to their corresponding relative chains, $\phi_{K_0}(c)=c+\mathsf{C}_d(K_0)$.
This homomorphism induces another homomorphism, 
$\phi_{K_0}^*:\mathsf{H}_d(K)\rightarrow \mathsf{H}_d(K,K_0)$,
mapping homology classes of $K$ to their corresponding
relative homology classes, $\phi_{K_0}^*(h)=\phi_{K_0}(z)+\mathsf{B}_d(K,K_0)$
for any $z\in h$.

Given a $d$-chain $c\in \mathsf{C}_d$, its corresponding relative chain $\phi_{K_0}(c)$
is a relative cycle if and only if $\partial_d(c)$ is carried by $K_0$. Furthermore,
it is a relative boundary if and only if there is a $(d+1)$-chain $c'\in \mathsf{C}_{d+1}(K)$
such that $c-\partial_{d+1}(c')$ is carried by $K_0$. 

These ideas are illustrated in Figure \ref{fig:relativeHom}.
Although $z_1$ and $z_2$ are both nonbounding cycles in $K$, $\phi_{K_0}(z_1)$ is a nonbounding
relative cycle whereas $\phi_{K_0}(z_2)$ is only a relative boundary. Although chains $c_1$ and $c_2$
are not cycles in $K$, $\phi_{K_0}(c_1)$ and $\phi_{K_0}(c_2)$ are relative cycles homologous to $\phi_{K_0}(z_1)$
and $\phi_{K_0}(z_2)$, respectively. 
\begin{figure}[hbtp]
    \centering
    \includegraphics[width=0.55\textwidth]{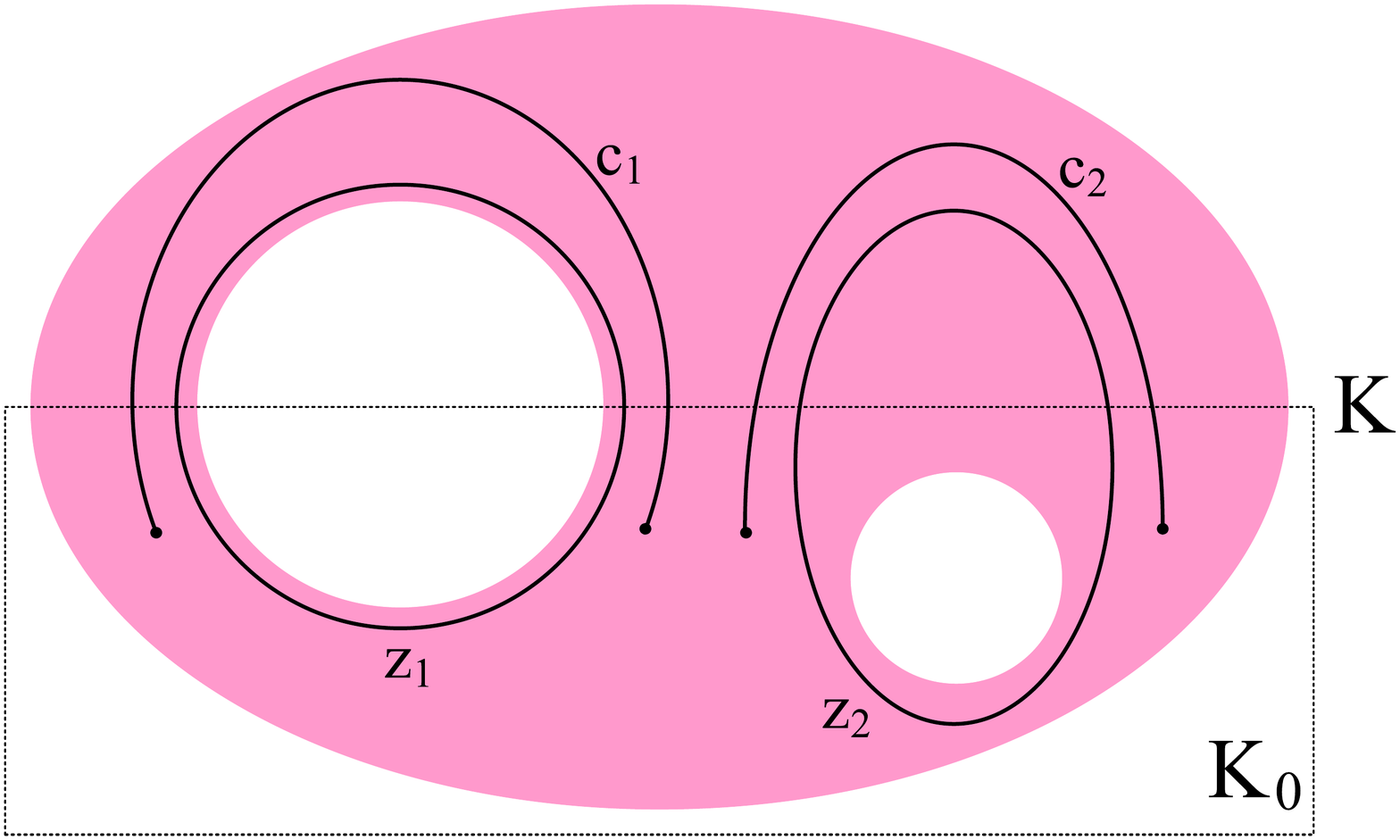}
    \caption{ A disk with two holes, whose triangulation is $K$. Simplices of $K$ lying completely
    in the dotted rectangle form a subcomplex $K_0$. The 1-dimensional relative homology group
    $\mathsf{H}_1(K,K_0)$ has dimension 1, although $\mathsf{H}_1(K)$ has dimension 2. The nontrivial class $[z_2]$ is carried by $K_0$.}
    \label{fig:relativeHom}
\end{figure}

Note that $[z_1]$ and $[z_2]$ are both nontrivial homology classes
in $K$. But their correspondences in the relative homology group may not necessarily be nontrivial.
We can see that $\phi_{K_0}^*([z_1])$ is a nontrivial relative homology class, whereas $\phi_{K_0}^*([z_2])$
is trivial. We say that the class $[z_2]$ is {\it carried by} $K_0$. This concept play an important role in our definition of the size measure. Further details will be given
in Section \ref{sec:sizeDef}.

%%%%%%%%%%%%%%%%%%%%%%%%%%%%%%%%%%%%%%%%%%%%%%%%%%

\subsection{Rank Computations of Sparse Matrices over Finite Fields}

Wiedemann \cite{Wiedemann86} presented a randomized algorithm to capture
the rank of a sparse matrix over finite field. His method performs a binary search
for the rank. For an $m\times n$ sparse matrix $A$, the algorithm starts with
$s=\min(m,n)/2$. It tests if $s>\rank(A)$ or not, and then decides whether
$s=s/2$ or $s=3s/2$. For each $s$, $s\times m$ and $s\times n$ matrices
$P$ and $Q$ are radomly generated for several times. If $PAQ$ is singular
all the times, $s>\rank(A)$ with high probability. The expected time of the
algorithm is $O(n(\omega+n\log n)\log n)$, where $n$ is the maximal dimension of the matrix and
$\omega$ is the total number of nonzero entries in $A$. 

%%%%%%%%%%%%%%%%%%%%%%%%%%%%%%%%%%%%%%%%%%%%%%%%%%%
%%%%%%%%%%%%%%%%%%%%%%%%%%%%%%%%%%%%%%%%%%%%%%%%%%%

\section{Defining the Problem}
In this section, we provide a technique for ranking homology
classes according to their importance. Specifically, we solve the three problems 
mentioned in Section \ref{sec:intro} by providing 
\begin{itemize}
\item a meaningful size measure for homology classes that is computable
in arbitrary dimension;  
\item localized cycles which are consistent with the size measure of 
their homology classes;
\item and an optimal homology basis which distinguishes large classes from
small ones effectively.
\end{itemize}

%%%%%%%%%%%%%%%%%%%%%%%%%%%%%%%%%%%%%%%%%%%%%%%%%%%
\subsection{The Discrete Geodesic Distance}

In order to measure the size of homology classes, we need a notion of distance.  As we will deal with a simplicial complex $K$, it is most natural to introduce a discrete metric, and corresponding distance functions.  We define the {\it discrete geodesic distance} from a vertex $p\in \vertex(K)$, 
$f_p:\vertex(K) \rightarrow \mathbb{Z}$, as follows. For any vertex
$q\in \vertex(K)$, $f_p(q)=\dist(p,q)$ is the length of the shortest path connecting
$p$ and $q$, in the $1$-skeleton of $K$; it is assumed that each edge length is one, though this can easily be changed.  We may then extend this distance function from vertices to higher dimensional simplices naturally.  For any simplex $\sigma \in K$, $f_p(\sigma)$ is the maximal
function value of the vertices of $\sigma$, $f_p(\sigma)=
\max_{q\in \vertex(\sigma)}f_p(q)$. 
Finally, we define a geodesic ball $B_p^r$, $p\in \vertex(K)$, $r\ge 0$, 
as the subset of $K$, $B_p^r = \{ \sigma \in K \mid f_p(\sigma) \le r \}$. It is straightforward to show that these subsets are in fact subcomplexes.

%%%%%%%%%%%%%%%%%%%%%%%%%%%%%%%%%%%%%%%%%%%%%%%%%%%

\subsection{Measuring the Size of a Homology Class}
\label{sec:sizeDef}

Using notions from relative homology, we proceed to define the size
of a homology class as follows. Given a simplicial complex $ K $,
assume we are given a collection of subcomplexes 
$\mathcal{L}=\{ L \subseteq  K \}$. Furthermore, each of
these subcomplexes is endowed with a size. In this case, we define the size of 
a homology class $h$ as the size of the smallest $ L $
carrying $h$. Here we say a subcomplex $ L $
{\it carries} $h$ if $h$ has a trivial image in the 
relative homology group $\mathsf{H}_d( K , L )$, namely,
$\phi_L^*(h)=\mathsf{B}_d( K , L )$. In Figure \ref{fig:relativeHom}, the
class $[z_2]$ is carried by $K_0$, whereas $[z_1]$ is not.
\begin{definition}
The size of a class $h$, $S(h)$, is the size of the smallest measurable subcomplex carrying $h$,
formally,
\begin{equation*}
S(h)=\min_{ L \in \mathcal{L}}\size( L )\quad s.t.\quad \phi_L^*(h)=\mathsf{B}_d( K , L ).
\end{equation*}
\end{definition}

To facilitate computation, we prove the following theorem.
\begin{theorem}
The size of a homology class $h$, is the size of the smallest measurable subcomplex 
carrying one of its cycles, $z\in h$, formally,
\begin{equation*}
S(h)=\min_{ L \in \mathcal{L}}\size( L )\quad s.t.\quad \exists z\in h: z\subseteq  L ,
\end{equation*}
\label{thm:size}
\end{theorem}
\begin{proof}
As we know, for any cycle $z\in h$, the relative chain $\phi_L(z)$
is a relative boundary if and only if there is a $(d+1)$-chain 
$c'\in \mathsf{C}_{d+1}( K )$ such that $z-\partial_{d+1}(c')$
is carried by $ L $. This means that $h$ is carried by $ L $ if and only if
there exists some cycle $z\in h$ carried by $ L $.
\end{proof}

In this paper, we take $\mathcal{L}$ to be the set of discrete geodesic balls,
$\mathcal{L}=\{B_p^r\mid p\in  \vertex(K) , r\geq 0\}$. 
The size of a geodesic ball is naturally its radius $r$. Combining the size
definition and the theorem we have just proven, we define the size measure of 
homology classes as follows.
\begin{definition}
The {\it size} of a homology class is the radius of the smallest
geodesic ball carrying one of its cycles, formally,
\begin{equation*}
S(h)=\min r\quad s.t.\quad \exists p\in  \vertex(K) \quad {\rm and}\quad z\in h:z\subseteq B_p^r.
\end{equation*}
This smallest geodesic ball is denoted as $B_{min}(h)$ for convenience,
whose radius is $S(h)$.
\label{def:size}
\end{definition}
In Figure \ref{fig:badGenerator} (right), the three geodesic balls
centered at $p_1$, $p_2$ and $p_3$ are the smallest geodesic
balls carrying nontrivial homology classes $[z_1]$, $[z_2]$ and $[z_3]$, respectively. 
Their radii are the size of the three classes. In Figure \ref{fig:tube}, the smallest geodesic ball 
carrying a nontrivial homology class is the pink one centered at $p_2$
\footnote{This geodesic ball actually carries the shortest cycle of the class
using the definition of Erickson and Whittlesey \cite{EricksonW05}. We will
discuss this in Section \ref{sec:lowDimResult}.}
, not the one centered at $p_1$. Note that these geodesic ball may not look like 
Euclidean balls in the embedding space. 

\begin{figure}[hbtp]
    \centering
    \includegraphics[width=0.6\textwidth]{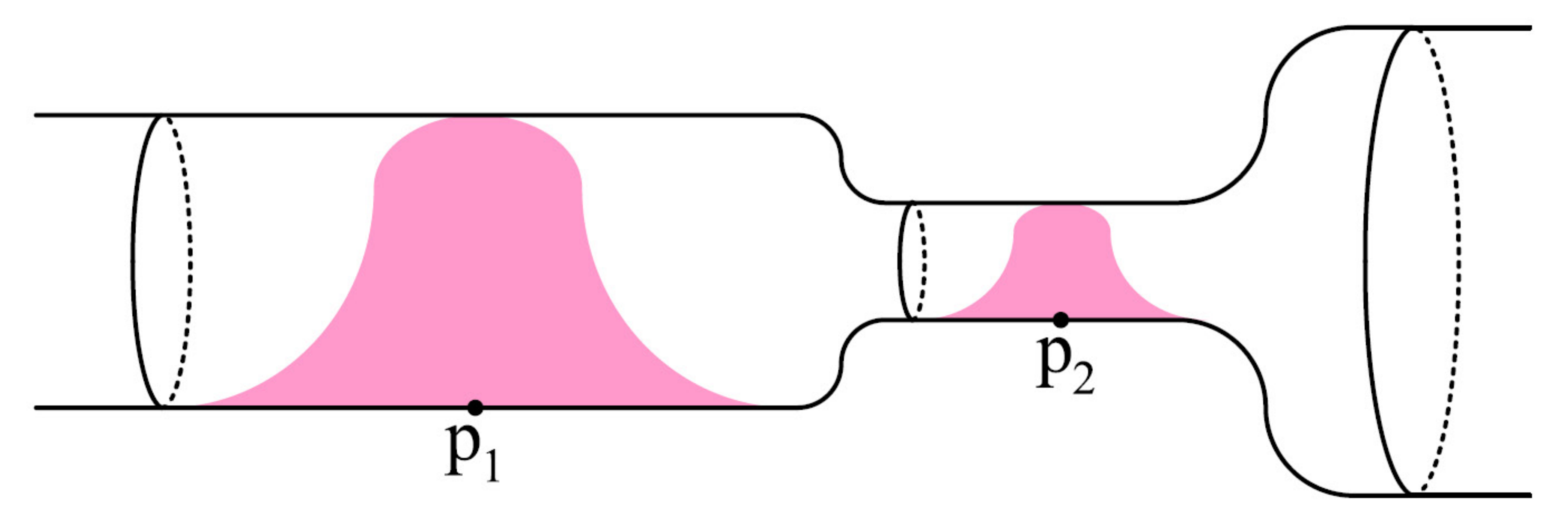}
    \caption{ On a tube, the smallest geodesic ball is centered at $p_2$, not $p_1$.}
    \label{fig:tube}
\end{figure}

%%%%%%%%%%%%%%%%%%%%%%%%%%%%%%%%%%%%%%%%%%%%%%%%%%%
\subsection{A Localized Cycle}
We would like to localize a homology class with a cycle which conveys 
its size. Define the {\it radius} of a cycle $z$ as,
\begin{eqnarray*}
\rad(z)=\min_{p\in \vertex( K) }\max_{q\in \vertex(z)} \dist(p,q),
\end{eqnarray*}
which is a natural extension of the canonical definition of radius, e.g.~of a Euclidean ball.
We define the {\it localized cycles} of a homology class $h$ as the one with the minimal
radius, namely, $z_0 = \argmin_{z\in h}\rad(z)$.

Based on Theorem \ref{thm:size}, it is not hard to see that the size of a class $h$ is 
equal to the minimal radius of its cycles, namely, $S(h) = \min_{z\in h}\rad(z)$, which
is exactly the radius of its localized cycles.
Thus, this definition of localized cycles agrees with our size measure for homology classes.

Given a homology class $h$, any of its cycles carried by $B_{min}(h)$ has the radius 
$S(h)$, and thus is localized.
In Figure \ref{fig:badGenerator}, $z_1$ and $z_2$ are localized cycles of 
$[z_1]$ and $[z_2]$ because they are carried by $B_{min}([z_1])$ and $B_{min}([z_1])$, 
respectively. 

\begin{remark}
Another quantity which can describe the size of a cycle is the {\it diameter}
\begin{eqnarray*}
\diam(z)=\max_{p,q\in \vertex(z)}\dist(p,q).
\end{eqnarray*}
We deliberately avoid this quantity because we conjecture computing
the cycle with the minimal diameter, $(\argmin_{z\in h}\diam(z))$,
is NP-complete. On the other hand, our definition of a localized cycle gives a 2-approximation
of the minimal diameter, formally,
\begin{eqnarray*}
\diam\left(\argmin_{z\in h}\rad(z)\right) \leq 2\min_{z\in h} \diam(z),
\end{eqnarray*}
which can be shown to be a tight bound.
\end{remark}

%%%%%%%%%%%%%%%%%%%%%%%%%%%%%%%%%%%%%%%%%%%%%%%%%%%

\subsection{The Optimal Homology Basis}

There are $2^{\beta_d}-1$ nontrivial homology classes. However, 
we only need $\beta_d$ of them to form a basis. The basis should
be chosen wisely so that we can easily distinguish important homology
classes from noise. See Figure \ref{fig:3circles} for an example.
There are $2^3-1=7$ nontrivial homology classes; we need three of
them to form a basis. We would prefer to choose $\{[z_1],[z_2],[z_3]\}$
as a basis, rather than $\{[z_1]+[z_2]+[z_3],[z_2]+[z_3],[z_3]\}$. 
The former indicates that there is one big cycle 
in the topological space, whereas the latter gives the impression of 
three large classes.

In keeping with this intuition, the {\it optimal homology basis}
is defined as follows.
\begin{definition}
The optimal homology basis is the basis for the homology group whose
elements' size have the minimal sum, formally,
\begin{equation*}
\mathcal{H}_d=\argmin_{\{h_1,...,h_{\beta_d}\}}\sum_{i=1}^{\beta_d} S(h_i),s.t.
\dim(\{h_1,...,h_{\beta_d}\})=\beta_d.
\end{equation*}
\label{def:optimalBasis}
\end{definition}
This definition guarantees that large homology classes appear
as few times as possible in the optimal homology basis. 
In Figure \ref{fig:3circles}, the optimal basis will be $\{ [z_1],  [z_2],  [z_3] \}$, which has only one large class.

%%%%%%%%%%%%%%%%%%%%%%%%%%%%%%%%%%%%%%%%%%%%%%%%%%%
%%%%%%%%%%%%%%%%%%%%%%%%%%%%%%%%%%%%%%%%%%%%%%%%%%%

\section{The Algorithm}
In this section, we introduce an algorithm to measure and localize the 
optimal homology basis as defined in
Definition \ref{def:optimalBasis}. We first introduce 
an algorithm to measure and localize the smallest 
homology class, namely, {\sf Measure-Smallest(K)}, which uses the 
persistent homology algorithm. Based on this procedure, 
we provide the algorithm {\sf Measure-All(K)}, which measures
and localizes the optimal 
homology basis. The algorithm takes $O(\beta_d^4n^4)$ time,
where $\beta_d$ is the Betti number and $n$ is the cardinality of
the input $K$. 

%%%%%%%%%%%%%%%%%%%%%%%%%%%%%%%%%%%%%%%%%%%%%%%%%%%
\subsection{Measuring and Localizing the Smallest Homology Class}
The procedure {\sf Measure-Smallest(K)} measures and localizes
the smallest nontrivial homology class, namely, the one with the smallest
size, 
\begin{eqnarray*}
h_{min}=\argmin_{h\in \mathsf{H}_d(K):h\neq \mathsf{B}_d(K)}S(h).
\end{eqnarray*}
The output of this procedure will be a pair $(S_{min},z_{min})$,
where $S_{min}=S(h_{min})$ and $z_{min}$ is a localized cycle
of $h_{min}$. According to the definitions, this pair is 
determined by the smallest geodesic ball carrying $h_{min}$,
namely, $B_{min}(h_{min})$. Once this ball is computed, its
radius is $S_{min}$, and a cycle of $h_{min}$
carried by this ball is $z_{min}$.

We first present an algorithm to compute the smallest
geodesic ball carrying $h_{min}$, i.e.~$B_{min}(h_{min})$. 
Second, we introduce the technique for finding $z_{min}$ from the computed ball.
The two corresponding procedures are {\sf Bmin} and {\sf Localized-Cycle}.
See Algorithm \ref{alg:smallestSlow} for
pseudocode of the procedure {\sf Measure-Smallest(K)}.
\begin{algorithm}[h!]
 \caption{ {\sf Measure-Smallest(K)}}
 \label{alg:smallestSlow}

 \begin{algorithmic}[1]
   \GOAL  measuring and localizing $h_{min}$.
 
    \INPUT $K$: the given simplicial complex.

    \OUTPUT $S_{min},z_{min}$:the size and a localized cycle of $h_{min}$.

		\STATE $(r_{min},p_{min})$ $=$ {\sf Bmin(K)}
		\STATE $S_{min}$ $=$ $r_{min}$
		\STATE $z_{min}$ $=$ {\sf Localized-Cycle($p_{min}$,$r_{min}$,$K$)}
 \end{algorithmic}
\end{algorithm}

\subsubsection{Computing $B_{min}(h_{min})$}
It is straightforward to see that $B_{min}(h_{min})$ is also
the smallest geodesic ball carrying any nontrivial homology class
of $K$. It can be computed by computing and comparing the smallest
geodesic balls centered at all vertices carrying nontrivial classes.
See Algorithm \ref{alg:Bmin} for the procedure.
\begin{algorithm}[h!]
 \caption{ {\sf Bmin(K)} }
 \label{alg:Bmin}

 \begin{algorithmic}[1]
   \GOAL computing $B_{min}(h_{min})$.

    \INPUT $K$: the given simplicial complex.

    \OUTPUT $p_{min}$, $r_{min}$:the center and radius of $B_{min}(h_{min})$.

		\STATE $r_{min}$ $=$ $+\infty$
    \FOR{$p\in \vertex(K)$}
			\STATE apply the persistent homology algorithm to $K$ with filter function $f_{p}$
			\STATE $r(p)=$birth time of the first essential homology class
			\IF{ $r(p)$ $<$ $r_{min}$ }
				\STATE $p_{min}$ $=$ $p$
				\STATE $r_{min}$ $=$ $r(p)$
			\ENDIF
    \ENDFOR
 \end{algorithmic}
\end{algorithm}

\begin{theorem}
Procedure {\sf Bmin(K)} computes $B_{min}(h_{min})$.
\end{theorem}
\begin{proof}
For each vertex $p$, we compute the smallest geodesic ball centered
at $p$ carrying any nontrivial homology class, namely, $B_p^{r(p)}$.
We apply the persistent homology algorithm to $K$ with the filter
function $f_p$. Notice that a geodesic ball $B_p^r$ is the sublevel set
$f_p^{-1}(-\infty, r]\subseteq K$. Nontrivial homology classes of $K$ are
essential homology classes in the persistent homology algorithm. 
(For clarity, in the rest of this paper, we may use ``essential
homology classes'' and ``nontrivial homology classes of $K$'' interchangable.)
Therefore, the birth
time of the first essential homology class is $r(p)$, and the 
subcomplex $f_{p}^{-1}(-\infty,r(p)]$ is $B_{p}^{r(p)}$.

When all the $B_p^{r(p)}$'s are computed, we compare their radii
and pick the smallest one as $B_{min}(h_{min})$. 
\end{proof}

Once $B_{min}(h_{min})$ is computed, its radius
is the size of $h_{min}$. Any cycle of $h_{min}$
carried by $B_{min}(h_{min})$ is a localized cycle of $h_{min}$.
Next, we explain how to compute one such localized cycle.

%%%%%%%%%%%%%%%%%%%%%%%%%%%%%%%%%%%%%%%%%%%%%%%%%%%
\subsubsection{Computing a Localized Cycle of $h_{min}$}

The procedure {\sf Localized-Cycle($p_{min}$,$r_{min}$,$K$)} computes
a localized cycle of $h_{min}$. 
We assume that $B_{min}(h_{min})$, 
the smallest geodesic ball carrying the smallest homology class, carries
exactly one nontrivial homology class, (i.e. $h_{min}$ itself). 
\footnote{\label{foot:hminUnique} This assumption may not necessarily be true.
It is possible that $B_{min}(h_{min})$ carries two or more nontrivial classes.
Suppose $p_{min}$ is the center of $B_{min}(h_{min})$.
Then the proof can be easily modified to deal with this case, by fixing an order on 
simplices with the same function value $f_{p_{min}}$, and simulating this
order on $f_{p_{min}}$, i.e. treating $f_{p_{min}}(\sigma_1)< f_{p_{min}}(\sigma_2)$
if $\sigma_1$ comes before $\sigma_2$ (even though 
$f_{p_{min}}(\sigma_1)=f_{p_{min}}(\sigma_2)$).}
Any cycle carried by this ball which is nonbounding in $K$ is a 
cycle of $h_{min}$, and thus is a localized cycle of $h_{min}$. 
Therefore, we first compute a basis for all the cycles carried by
$B_{min}(h_{min})$. Second, we check elements in this basis one
by one until we find one which is nonbounding in $K$.
See Algorithm \ref{alg:Localized-Cycle} for the procedure.
Note that we use the algorithm of Wiedemann \cite{Wiedemann86} for rank computation,
because the related matrices are sparse.
\begin{algorithm}[h!]
 \caption{ {\sf Localized-Cycle($p_{min}$,$r_{min}$,$K$)} }
 \label{alg:Localized-Cycle}

 \begin{algorithmic}[1]
   \GOAL compute a localized cycle of $h_{min}$.

    \INPUT  $p_{min}$,$r_{min}$: the center and radius of $B_{min}(h_{min})$.\\
    				$K$: the given simplicial complex.

    \OUTPUT $z_{min}$: a localized cycle of $h_{min}$.

		\STATE $rank_0$ $=$ $\rank(\partial_{d+1})$
		\STATE construct $\partial_{d}'$ by picking columns of $\partial_{d}$ whose corresponding simplices belong to $B_{min}(h_{min})$
		\STATE reduce $\partial_{d}'$ and get $R$ and $V$
		\FOR{$z$ $=$ columns in $V$ corresponding to zero columns in $R$}
			\STATE $rank_1$ $=$ $\rank([z,\partial_{d+1}])$
			\IF{ $rank_1$ $\neq$ $rank_0$ }
				\STATE $z_{min}$ $=$ $z$
				\STATE break
			\ENDIF
    \ENDFOR
 \end{algorithmic}
\end{algorithm}

\begin{theorem}
The procedure {\sf Localized-Cycle($p_{min}$,$r_{min}$,$K$)} computes a localized cycle of $h_{min}$.
\end{theorem}
\begin{proof}
The cycles carried by $B_{min}(h_{min})$ form a vector space
\begin{eqnarray*}
\mathsf{Z}_d(K)\cap \mathsf{C}_d(B_{min}(h_{min})).
\end{eqnarray*}
We compute its basis by column reducing the boundary matrix restricted to
$B_{min}(h_{min})$. After the reduction, each zero column corresponds to an
element of the basis. More specifically, we compute the basis as follows.
We first construct a matrix $\partial_d'$ with columns of the boundary
matrix $\partial_d$ whose corresponding simplices belong to 
$B_{min}(h_{min})$. Next we perform a column reduction on this matrix from left
to right, like in the persistent homology algorithm. The reduction corresponds
to a matrix multiplication 
\begin{eqnarray*}
R=\partial_d'V,
\end{eqnarray*}
where $R$ is the reduced matrix and $V$ is an upper triangular matrix.
The columns in $V$ corresponding to zero columns in $R$ form the basis
of cycles carried by $B_{min}(h_{min})$.

Next, we check elements in this basis one by one to find one which is
nonbounding in $K$. An element of this basis, $z$, is nonbounding
in $K$ if and only if it cannot be expressed as a linear combination 
of boundaries of $K$. Since columns of the boundary matrix 
$\partial_{d+1}$ generate $\mathsf{B}_d(K)$, we just
need to compute the rank of the matrix $[z,\partial_{d+1}]$ and
compare it with the rank of $\partial_{d+1}$. The cycle $z$ is nonbounding in $K$
if and only if these two ranks are different.
\end{proof}
%%%%%%%%%%%%%%%%%%%%%%%%%%%%%%%%%%%%%%%%%%%%%%%%%%%

\subsection{The Optimal Homology Basis}

In this section, we present the algorithm for computing the optimal homology basis defined in Definition \ref{def:optimalBasis}, namely, $\mathcal{H}_d$. We first show that the optimal homology basis
can be computed in a greedy manner. Second, we introduce an efficient greedy algorithm.

%%%%%%%%%%%%%%%%%%%%%%%%%%%%%%%%%%%%%%%%%%%%%%%%%%%

\subsubsection{Computing $\mathcal{H}_d$ in a Greedy Manner}
Recall that the optimal homology basis is 
\begin{equation*}
\mathcal{H}_d=\argmin_{\{h_1,...,h_{\beta_d}\}}\sum_{i=1}^{\beta_d} S(h_i) \, s.t. \,
\dim(\{h_1,...,h_{\beta_d}\})=\beta_d.
\end{equation*}  
We use matroid theory \cite{CormenLRC01} to show that 
we can compute the optimal homology basis with a greedy method.
Let $H$ be the set of nontrivial $d$-dimensional homology classes (i.e. the homology group minus the trivial class). 
Let $L$ be the family of sets of linearly independent nontrivial homology classes.
Then we have the following theorem.  The same 
result has been mentioned in \cite{EricksonW05}.
\begin{theorem}
The pair $(H,L)$ is a matroid when $\beta_d>0$.
\label{thm:matroid}
\end{theorem}
\begin{proof}
We show $(H,L)$ is a matroid by proving the following properties.
\begin{enumerate}
\item The set $H$ is finite and nonempty as $\card(H)=2^{\beta_d}-1$.

\item For any set of linearly independent nontrivial homology classes, its subsets
are also linearly independent. Therefore, elements in $L$ are independent subsets of $H$,
and $L$ is hereditary.

\item For any two sets of linearly independent classes $l_1,l_2\in L$ such that
$\card(l_1)<\card(l_2)$, we can always find a homology class $h\in l_2\backslash l_1$ such
that $l_1\cup\{h\}$ is still linearly independent. Otherwise, any element in
$l_2$ is dependent on $l_1$. This means 
\begin{equation*}
\dim(l_2)\leq\dim(l_1)=\card(l_1)<\card(l_2),
\end{equation*}
which contradicts the linear independence of $l_2$. Therefore, $(H,L)$
satisfies the exchange property.
\end{enumerate}
\end{proof}

We construct a weighted matroid by assigning each nontrivial homology 
class its size as the weight. This weight function is strictly positive
because a nontrivial homology class can not be carried by a geodesic ball
with radius zero. According to matroid theory, we can compute the optimal homology basis
\begin{equation*}
\mathcal{H}_d=\argmin_{l\in L}\sum_{h\in l} S(h).
\end{equation*}
with a naive greedy method as follows. 
\begin{enumerate}
\item Sort elements in $H$ into
an order which is monotonically increasing according to size, namely, 
\begin{eqnarray*}
seq(H)&=(h_1,h_2,...,h_{(2^{\beta_d}-1)}),h_i\in H,\\
&\text{such that} \quad S(h_i)\leq S(h_j) \quad \forall i<j.
\end{eqnarray*}

\item Repeatedly pick the smallest class from $seq(H)$ that is linearly independent of those 
we have already picked, until no more elements are qualified.

\item The selected $\beta_d$ classes $\{h_{i_1},h_{i_2},...,h_{i_{\beta_d}}\}$ form the optimal
homology basis $\mathcal{H}_d$.  (Note that the $h$'s are ordered by size, i.e. $S(h_{i_k}) \le S(h_{i_{k+1}})$.)
\end{enumerate}

However, we cannot compute the exponentially long sequence $seq(H)$ (exponential in $\beta_d$) directly. 
Next, we present our greedy algorithm which is polynomial.

%%%%%%%%%%%%%%%%%%%%%%%%%%%%%%%%%%%%%%%%%%%%%%%%%%%

\subsubsection{Computing $\mathcal{H}_d$ with a Sealing Technique}

In this section, we introduce the algorithm for computing $\mathcal{H}_d$.
Instead of computing the exponentially long sequence $seq(H)$ directly, our
algorithm uses a sealing technique and takes time polynomial in $\beta_d$.

We start by measuring and localizing the smallest 
homology class of the given simplicial complex $K$, which is also the first
class we choose for $\mathcal{H}_d$. We destroy this class by 
sealing up one of its cycles -- i.e.~the localized cycle we 
computed -- with new simplices. Next, we measure and localize the
smallest homology class of the augmented simplicial complex $K'$. 
This class is the second smallest
homology class in $\mathcal{H}_d$. We destroy this class again and proceed
for the third smallest class in $\mathcal{H}_d$.  This process is repeated for $\beta_d$
rounds, yielding $\mathcal{H}_d$.

We destroy a homology class by sealing up the class's localized cycle, which we have computed.
To seal up this cycle $z$, we add (a) a new vertex $v$; (b) a $(d+1)$-simplex for each $d$-simplex of $z$, with vertex set equal to the vertex set of the $d$-simplex together with $v$; (c) all of the faces
of these new simplices. In Figure \ref{fig:mAll},
a $1$-cycle with four edges, $z_1$, is sealed up with one new vertex, 
four new triangles and four new edges.

We assign the new vertices $+\infty$ geodesic distance from 
any vertices with which they share an edge in the original complex $K$. Whenever we run
the persistent homology algorithm, all of the new simplices have $+\infty$ 
filter function values. Furthermore, in the procedure {\sf Measure-Smallest($K'$)}, we will not consider
any geodesic ball centered at these new vertices. In other words, the geodesic distance from these new 
vertices will never be used as a filter function.
Algorithm \ref{alg:measureAll} contains the pseudocode.
\begin{algorithm}[h!]
 \caption{ {\sf Measure-All($K$)} }
 \label{alg:measureAll}

 \begin{algorithmic}[1]
    \GOAL compute the optimal homology basis, $\mathcal{H}_d$.

    \INPUT $K$: the given simplicial complex.

    \OUTPUT $\mathcal{H}_d$: the optimal homology basis.

		\STATE $K'$ $=$ $K$
		\STATE $\mathcal{H}_d$ $=$ $\emptyset$
		\FOR{ i $=$ $1$ to $\beta_d$ }
				\STATE $h=(S,z)=${\sf Measure-Smallest($K'$)}
				\STATE $\mathcal{H}_d$ $=$ $\mathcal{H}_d$ $\cup$ $\{h\}$
				\STATE seal $z$ with new simplices, augment $K'$ accordingly
				\STATE $\forall \sigma\in K'\backslash K,p\in K,f_p(\sigma)=+\infty$
		\ENDFOR
 \end{algorithmic}
\end{algorithm}

Next, we prove that this algorithm does compute the optimal homology basis
$\mathcal{H}_d$. We will prove in Theorem \ref{thm:sealingUp} that {\sf Measure-All($K$)} produces
the same result as the naive greedy method presented in the previous
section. We begin by proving a lemma, based on the assumption in Footnote
\ref{foot:hminUnique} that $h_{min}$ is the only notrivial homology class
carried by $B_{min}(h_{min})$.
\begin{lemma}
Given a simplicial complex $K$, if we seal up its smallest homology class
$h_{min}(K)$, any other nontrivial homology class of $K$, $h$, is still
nontrivial in the augmented simplicial complex $K'$. In other words,
any cycle of $h$ is still nonbounding in $K'$.
\label{lem:sealingUp}
\end{lemma}
\begin{proof}
As we deal with two complexes $K$ and $K'$ with $K \subseteq K'$, we let $I:\mathsf{C}_d(K) \to \mathsf{C}_d(K')$ and $I^*:\mathsf{H}_d(K) \to \mathsf{H}_d(K')$ be the maps induced by inclusion.  Also, for a chain $c$, let $|c|$ be the simplicial complex composed of simplices from $c$ and their faces.

We proceed by contradiction.
Let $z_{min}\in h_{min}(K)$ be the localized cycle of $h_{min}(K)$ that we
seal up. For any nontrivial class $h\in \mathsf{H}_d(K)$, $h\neq h_{min}(K)$, 
suppose $I^*(h)$ is trivial.  We will show that there exists a cycle in $h$ which is carried by $B_{min}(h_{min})$, which contradicts the fact that $h_{min}$ is the only nontrivial class carried by $B_{min}(h_{min})$.

Suppose $I^*(h)$ is trivial. For any cycle $z\in h$, its corresponding 
$I(z)$ is the boundary of a $(d+1)$-chain in $K'$.  As $z$ is nonbounding in $K$, it must be the case that at least one of the simplices of this $(d+1)$-chain must be new.  That is 
\begin{eqnarray*}
I(z)=\partial_{d+1}\left(\sum_{\sigma\in K'\backslash K}a_{\sigma}\sigma+\sum_{\tau\in K}a_{\tau}\tau\right),
\end{eqnarray*}
where at least one $a_{\sigma}\neq 0$. But there exists a cycle $z'$ which is homologous
to $z$ in $K$, with $z'=z-\partial_{d+1}(\sum_{\tau\in K}a_{\tau}\tau)$, which yields, finally,
that $I(z')=\partial_{d+1}(\sum_{\sigma\in K'\backslash K}a_{\sigma}\sigma)$.  In other words, $I(z')$ is the boundary of a $(d+1)$-chain all of whose simplices are new.
Any simplex of $|I(z')|$ is a face of the new simplices and belongs to the original complex $K$,
and thus belongs to $|I(z_{min})|$.
It follows that $I(z')$ is carried by the simplicial complex corresponding to $I(z_{min})$,
$|I(z_{min})|$; and hence, $z'$ is carried by $|z_{min}|$.
Consequently, $z'$ and $h$ are carried by $B_{min}(h_{min})$, which leads to the desired contradiction.
\end{proof}

\begin{theorem}
The procedure {\sf Measure-All($K$)} computes $\mathcal{H}_d$.
\label{thm:sealingUp}
\end{theorem}
\begin{proof}
We prove the theorem by showing that the sealing up technique produces 
the same result as the naive greedy algorithm, namely,
$\mathcal{H}_d=\{h_{i_1},h_{i_2},...,h_{\beta_d}\}$. We show that for any $l\leq \beta_d$, after computing
and sealing up the first $l-1$ classes of $\mathcal{H}_d$, i.e. $\{h_{i_1},...,h_{i_{l-1}}\}$, the next class
we choose is exactly $h_{i_l}$.  In other words, the localized
cycle and size of the smallest class of the augmented simplicial complex
$K^{l-1}$ are equal to that of $h_{i_{l}}$. 

First, any class between $h_{i_{l-1}}$ and $h_{i_l}$ in $seq(H)$ will not be
chosen. Any such class $h_j$ is linearly dependent on classes that have already
been chosen, namely, $\{h_{i_1},...,h_{i_{l-1}}\}$. Since these classes
have been sealed up, a cycle of $h_j$ is a boundary in $K^{l-1}$. Thus, $h_j$
cannot be chosen.

Second, Lemma \ref{lem:sealingUp} leads to the fact that for any class in $seq(H)$
that is not linearly dependent on $\{h_{i_1},...,h_{i_{l-1}}\}$, it is nontrivial in $K^{l-1}$.

Third, the smallest class of $K^{l-1}$, $h_{min}(K^{l-1})$, corresponds to $h_{i_{l}}$: any new
simplex belonging to $K^{l-1}\backslash K$ will not change the computation of the geodesic balls $B_p^r$
with finite radius $r$, and thus will change neither the size measurement nor the localization.  Thus, the $h_{min}(K^{l-1})$ computed by the sealing technique is identical to $h_{i_l}$ computed by the naive greedy method,
in terms of the size and the localized cycle.
\end{proof}
The algorithm is illustrated in Figure \ref{fig:mAll}. The rectangle, $z_1$, and the octagon, $z_2$,
are the localized cycles of the smallest and the second smallest homology 
classes ($S([z_1])=2$,$S([z_2])=4$).
The nonbounding cycle $z_3=z_1+z_2$ corresponds to the largest nontrivial 
homology class $[z_3]=[z_1]+[z_2]$ ($S([z_3])=5$). After the first round, 
we choose $[z_1]$ as the smallest class in $\mathcal{H}_1$.
Next, we destroy $[z_1]$ by sealing up $z_1$, which yields the augmented 
complex $K'$. This time, we choose $[z_2]$, giving $\mathcal{H}_1=\{[z_1],[z_2]\}$.
\begin{figure}[hbtp]
    \centerline{
    \begin{tabular}{cc}
		\includegraphics[width=0.42\textwidth]{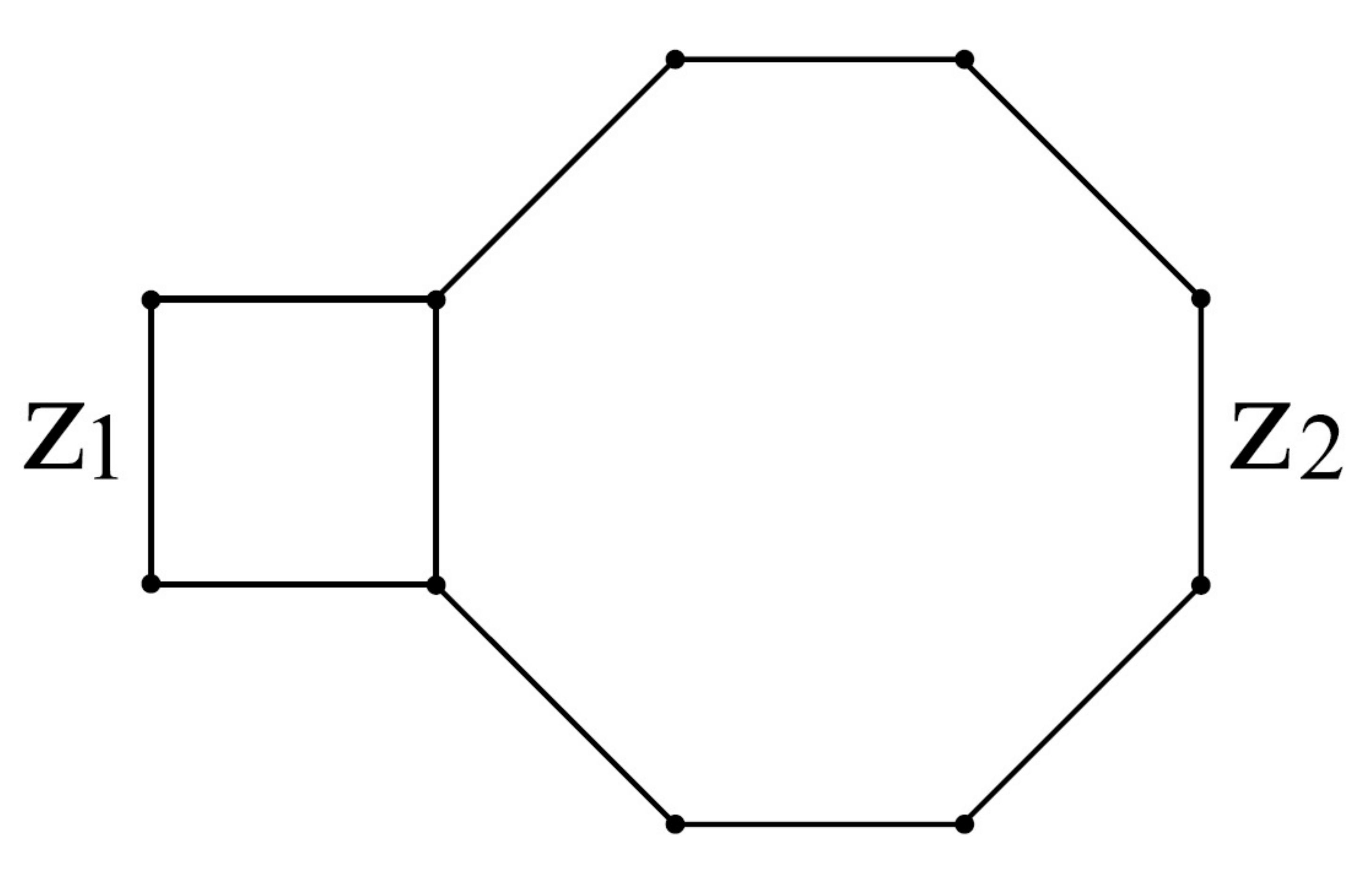} &
		\includegraphics[width=0.42\textwidth]{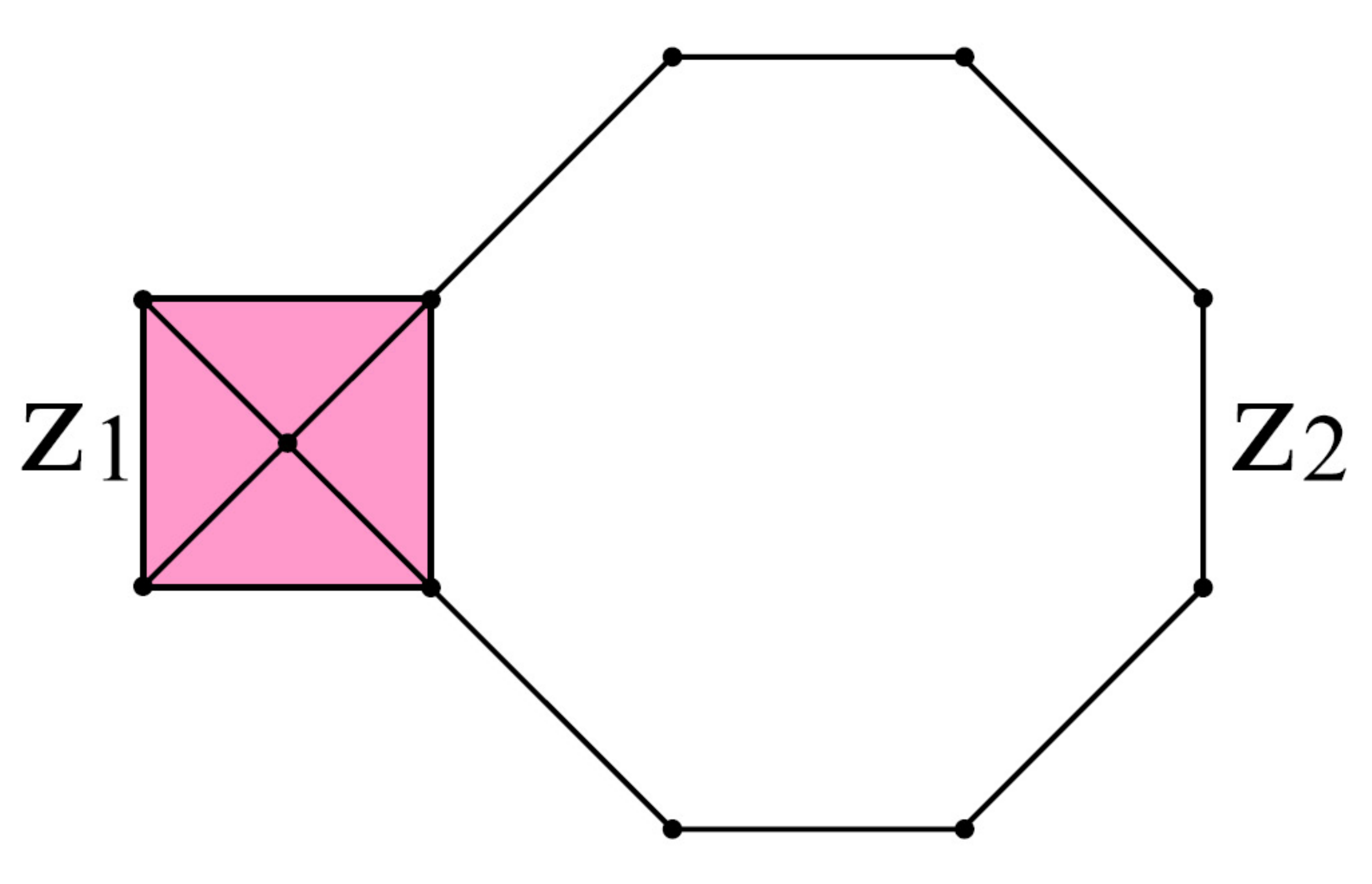}
    \end{tabular}}
    \caption{ Left: the original complex $K$. Right: the augmented complex $K'$ after sealing up the smallest class, $[z_1]$. }
    \label{fig:mAll}
\end{figure}

\subsection{Complexity}
We analyze the complexity of the non-refined algorithm.
Denote $n$ and $m$ as the upper bounds of the total numbers of simplices
of the original complex $K$ and the intermediate complex $K'$, respectively.
The algorithm runs the procedure {\sf Measure-Smallest} $\beta_d$ times with the input
$K'$, and thus runs the procedures {\sf Bmin} and {\sf Localized-Cycle} $\beta_d$ times with
the input $K'$. 

The procedure {\sf Bmin} runs the persistent homology algorithm
on the intermediate complex, $K'$, using filter function $f_p$ for each vertex of the original
complex, $K$. Therefore, each time {\sf Bmin} is called, it takes $O(nm^3)$ time.

The procedure {\sf Localized-Cycle} runs the persistent homology algorithm once, and
Wiedemann's rank computation algorithm $O(m)$ times. The matrices used for rank computations
are $[z,\partial_{d+1}]$ which have $O(m)$ nonzero entries. Therefore, 
each time {\sf Localized-Cycle} is called, it takes $O(m^3\log^2m)$ time.

In total the whole algorithm takes $O(\beta_d(nm^3+m^3\log^2m))=O(\beta_dnm^3)$ time.
Next, we bound $m$, the size of the intermediate simplicial complex $K'$.
During the algorithm, we seal up $\beta_d$ nonbounding cycles. 
For each sealing, the number of newly added simplices is bounded by the number of simplices of the 
sealed cycle. As we have shown, each cycle we seal up only contains simplices in the 
original complex $K$. Therefore, the number of new simplices used to seal up each cycle
is $O(n)$. The size of the intermediate simplicial complex, $K'$,
is $O(\beta_dn)$ throughout the whole algorithm.

Finally, substitute $\beta_dn$ for $m$. We conclude that the algorithm
takes $O(\beta_dnm^3)=O(\beta_dn(\beta_dn)^3)=O(\beta_d^4n^4)$ time.

%%%%%%%%%%%%%%%%%%%%%%%%%%%%%%%%%%%%%%%%%%%%%%%%%%%
%%%%%%%%%%%%%%%%%%%%%%%%%%%%%%%%%%%%%%%%%%%%%%%%%%%

\section{An Improvement Using Finite Field Linear Algebra}
In this section, we present an improvement on the algorithm
presented in the previous section, more specifically, an improvement
on the procedure {\sf Bmin($K$)}. The idea is based on
the finite field linear algebra behind the homology.

We first observe that for neighboring vertices, $p_1$ and $p_2$,
the persistence diagrams using $f_{p_1}$ and $f_{p_2}$ as filter
functions are close. In Theorem \ref{thm:neighborClose}, we prove
that the birth times of the first essential homology classes using
$f_{p_1}$ and $f_{p_2}$ differ by no more than $1$. This observation
suggests that for each $p$, instead of computing $B_{p}^{r(p)}$
we may just test whether a certain geodesic ball carries any essential
homology class.
Second, with some algebraic insight, we reduce the problem
of testing whether a geodesic ball carries any essential homology class
to the problem of comparing dimensions of two vector spaces.
Furthermore, we use Theorem \ref{thm:rank} to reduce the problem to
rank computations of sparse matrices on the $\mathbb{Z}_2$ field,
for which we have ready tools (of Wiedemann \cite{Wiedemann86}).  

In doing so, we improve the complexity of computing the optimal homology
basis to $O(\beta_d^4 n^3\log^2)$.
\begin{remark} 
This complexity is close to that of the persistent homology algorithm, whose
complexity is $O(n^3)$. Given the nature of the problem, it seems likely 
that the persistence complexity is a lower bound. If this is the case, the current
algorithm is nearly optimal.
\end{remark}

\begin{remark} 
Cohen-Steiner et al.~\cite{Cohen-SteinerEM06}
provided a linear algorithm to maintain the persistent diagram while changing
the filter function. However, this algorithm is not directly applicable in our context.
The reason is that it takes $O(n)$ time to update
the persistent diagram for a transposition in the simplex-ordering.
In our case, even for filter functions of two neighboring 
vertices, it may take $O(n^2)$ transpositions to transform one simplex-ordering
into the other.
Therefore, updating the persistent diagram while changing the filter function
takes $O(n^2)\times O(n)=O(n^3)$ time. This is the same amount of time it
would take to compute the persistent diagram from scratch.
\end{remark}

In this section, we assume that $K$ has a single component; 
multiple components can be accommodated with a simple modification.
For convenience, we use ``carrying nonbounding cycles'' and
``carrying essential homology classes'' interchangeably, because a geodesic ball
carries essential homology classes of $K$ if and only if it carries nonbounding
cycles of $K$.

%%%%%%%%%%%%%%%%%%%%%%%%%%%%%%%%%%%%%%%%%%%%%%%%%%%

\subsection{The Stability of Persistence Leads to An Improvement}
Cohen-Steiner et al.~\cite{Cohen-SteinerEH07} proved that the change, suitably defined, of the persistence
of homology classes is bounded by the changes of the filter functions.
Since the filter functions of two neighboring vertices, $f_{p_1}$ and $f_{p_2}$,
are close to each other, the birth times of the first nonbounding cycles in
both filters are close as well. This leads to Theorem \ref{thm:neighborClose}.
\begin{theorem}
If two vertices $p_1$ and $p_2$ are neighbors, the birth times of the first
nonbounding cycles for filter functions $f_{p_1}$ and 
$f_{p_2}$ differ by no more than 1.
\label{thm:neighborClose}
\end{theorem}
\begin{proof}
We first prove that the filter functions are close for two neighboring 
vertices $p_1$ and $p_2$, formally, 
\begin{equation}
|f_{p_1}-f_{p_2}|_{\infty}\leq 1. 
\label{eqn:filterClose}
\end{equation}
For any vertex $q$, we can connect $q$ and $p_2$ by concatenating the edge 
$(p_1,p_2)$ to the shortest path connecting $q$ and $p_1$. Therefore the
geodesic distance between $q$ and $p_2$ is no greater than one plus the
geodesic distance between $q$ and $p_1$, formally,
\begin{eqnarray*} 
f_{p_2}(q)\leq 1+f_{p_1}(q).
\end{eqnarray*} 
It is trivial to see that we can switch $p_1$ and $p_2$ in this equation. 
Therefore, we have
\begin{eqnarray*}
|f_{p_1}(q)-f_{p_2}(q)|\leq 1.
\end{eqnarray*}
It is not hard to extend this equation from any vertex $q\in \vertex(K)$ to any simplex 
$\sigma\in K$. Therefore, Equation (\ref{eqn:filterClose}) is proven.

Next, we show that the birth times of the first nonbounding cycles in
the two filter functions are close, formally, 
\begin{equation}
|f_{p_1}(z')-f_{p_2}(z'')|\leq 1,
\label{eqn:firstUBCycleClose}
\end{equation} 
where $z'$ and $z''$ are the first
nonbounding cycles in the filters $f_{p_1}$ and $f_{p_2}$, respectively.
Here by slightly abusing the notation, we denote $f(z)$ as the birth 
time of the cycle $z$ in the filter $f$. 

It is not hard to see that
the birth time of any cycle $z$ is the maximum of the 
function values of its simplices, and thus, is the maximum of the 
function values of its vertices, formally, 
\begin{eqnarray*}
f(z)=\max_{q\in \vertex(z)} f(q).
\end{eqnarray*}
We prove Equation (\ref{eqn:firstUBCycleClose}) by contradiction. Suppose
\begin{eqnarray*}
f_{p_1}(z')-f_{p_2}(z'')\geq 2.
\end{eqnarray*}
We know that for any vertex $q\in \vertex(z'')$,
\begin{eqnarray*}
f_{p_2}(q)\leq f_{p_2}(z'')\leq f_{p_1}(z')-2.
\end{eqnarray*}
From Equation (\ref{eqn:filterClose}), we have 
\begin{eqnarray*}
&&f_{p_1}(q)\leq f_{p_2}(q)+1 \leq f_{p_1}(z')-1,\forall q\in \vertex(z''),\\
&\Rightarrow&f_{p_1}(z'')=\max_{q\in \vertex(z'')}f_{p_1}(q)\leq f_{p_1}(z')-1.
\end{eqnarray*} 
This contradicts the fact
that $z'$ is the first nonbounding cycle in the filter $f_{p_1}$. 
Therefore, the assumption is wrong, and
\begin{eqnarray*}
f_{p_1}(z')-f_{p_2}(z'')\leq 1.
\end{eqnarray*}
Similarly, we can prove that 
\begin{eqnarray*}
f_{p_2}(z'')-f_{p_1}(z')\leq 1.
\end{eqnarray*}
In summary, we have proven Equation (\ref{eqn:firstUBCycleClose}), and
consequently, proven the theorem.
\end{proof}

This theorem suggests a way to avoid computing $B_{p}^{r(p)}$
for all $p\in K$.
Recall that $r(p)$ is the radius of the smallest geodesic ball
centered at $p$ that carries any nonbounding cycle. 
Based on this theorem, we know that for any vertex $p_i$,
$r(p_i)\geq r(p_j)-1$ for any neighbor $p_j$.
Since our objective is to find the minimum of the $r(p)$'s, we can do a 
breadth-first search through all the vertices with 
global variables $r_{min}$ recording the smallest $r(p)$ we have found, \
and $p_{min}$ recording the corresponding center $p$. 

We start by applying the persistent homology algorithm on $K$ with
filter function $f_{p_0}$. Initialize $r_{min}$ as the birth time
of the first nonbounding cycle of $K$, $r(p_0)$, and $p_{min}$ as $p_0$. 
Next, we do a breadth-first
search through the rest vertices. For each vertex
$p_i, i\neq 0$, we know there exists a neighbor $p_j$ such that $r(p_j)\geq r_{min}$. 
Therefore, 
\begin{eqnarray*}
r(p_i)\geq r(p_j)-1\geq r_{min}-1.
\end{eqnarray*} 
We only need to test whether the geodesic ball
$B_{p}^{r_{min}-1}$ carries any nonbounding cycle of $K$. If so, $r_{min}$ is decremented
by one, and $p_{min}$ is updated to $p$. 

However, testing whether the subcomplex $B_{p}^{r_{min}-1}$ carries any
nonbounding cycle of $K$ is not as easy as computing nonbounding cycles 
of the subcomplex. A nonbounding cycle of $B_{p}^{r_{min}-1}$
may not be nonbounding in $K$ as we require.
For example, in Figure \ref{fig:torusWithTail}, we want to compute the smallest geodesic ball
centered at $p$ carrying any nonbounding cycle of $K$, $B_p^{r(p)}$.
The gray geodesic ball in the first figure does not carry any nonbounding
cycle of $K$, although it carries its own nonbounding cycles. The geodesic ball
in the second figure carries nonbounding cycles of $K$ and is the ball we want,
namely, $B_p^{r(p)}$. 
\begin{figure}[hbtp]
    \centerline{
    \begin{tabular}{lr}
		\includegraphics[width=0.2\textwidth]{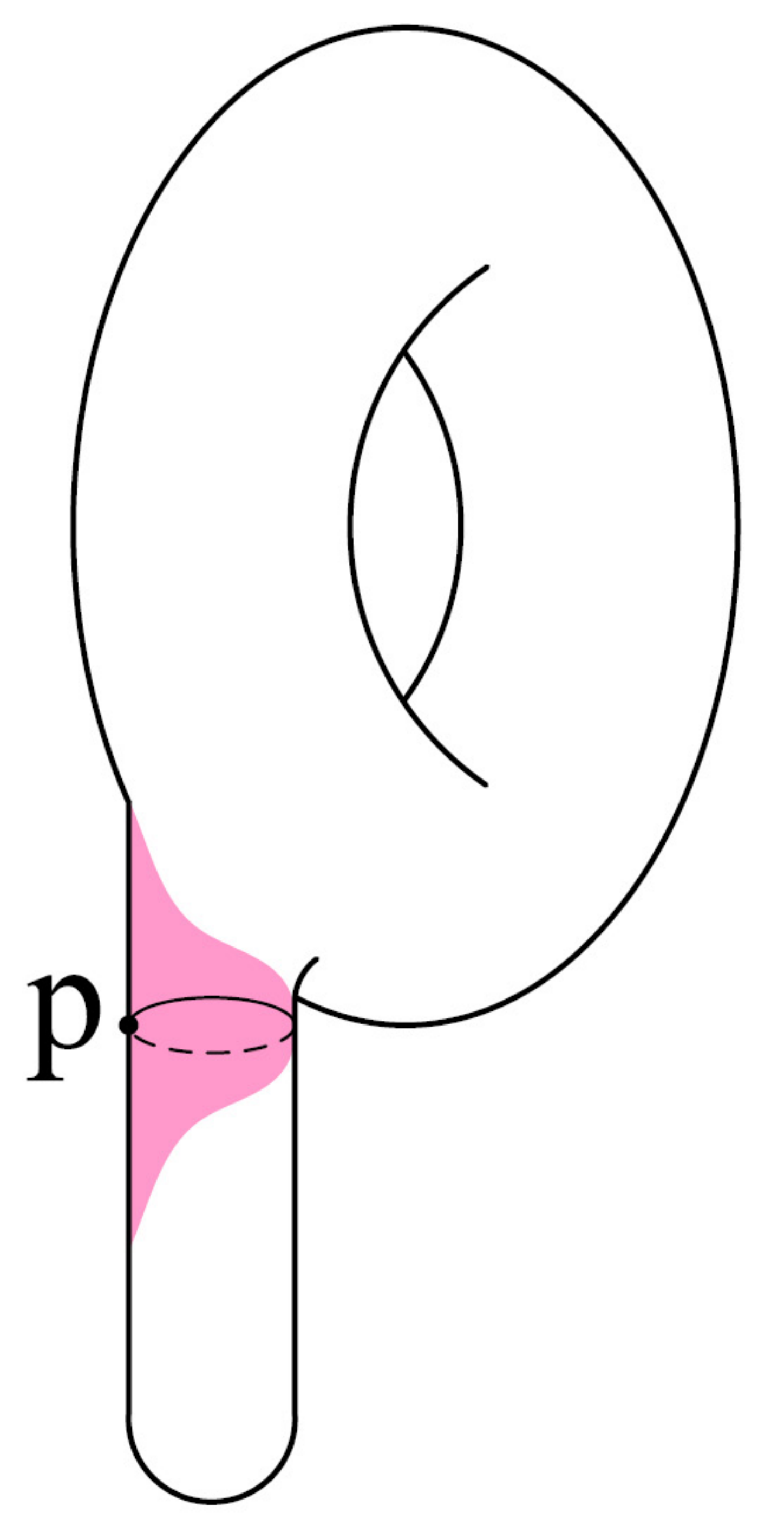} 
		\hspace{0.1in}&\hspace{0.1in}
		\includegraphics[width=0.2\textwidth]{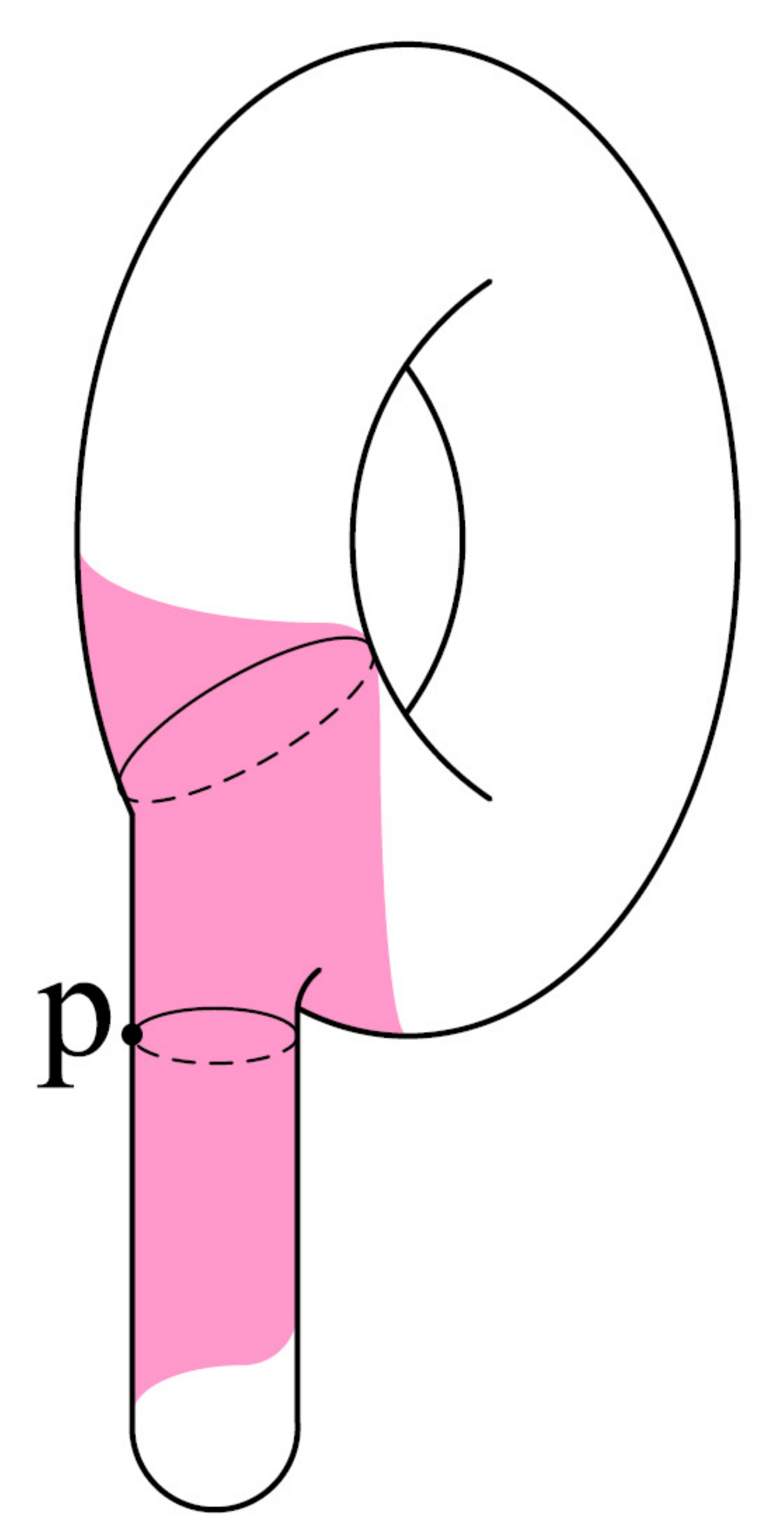} 
    \end{tabular}}
    \caption{ Computing $B_p^{r(p)}$ in a torus with tail. The ball in the
    second figure is what we want, although the one in the first figure has
    nontrivial topology.}
    \label{fig:torusWithTail}
\end{figure}
Therefore, we need algebraic tools to distinguish nonbounding cycles
of $K$ from those of the subcomplex $B_{p}^{r_{min}-1}$.
%%%%%%%%%%%%%%%%%%%%%%%%%%%%%%%%%%%%%%%%%%%%%%%%%%%

\subsection{Testing Whether a Subcomplex Carries Nonbounding Cycles of $K$}
In this subsection, we present the procedure for testing whether 
a subcomplex $K_0$ carries any nonbounding cycle of $K$. 
A chain in $K_0$ is a cycle if and only if it is a cycle of $K$.
However, solely from $K_0$, we are not able to tell whether a
cycle carried by $K_0$ bounds or not in $K$.
Instead, we write the set of cycles carried by $K_0$, 
$\mathsf{Z}_d^{K_0}(K)$, and the set
of boundaries of $K$ carried by $K_0$, $\mathsf{B}_d^{K_0}(K)$,
as sets of linear combinations with certain constraints. Consequently, 
we are able to test 
whether any cycle carried by $K_0$ is nonbounding in $K$ by comparing
the dimensions of $\mathsf{Z}_d^{K_0}(K)$ and $\mathsf{B}_d^{K_0}(K)$.
Theorem \ref{thm:rank} shows that these dimensions can be computed by
rank computations of sparse matrices.

\subsubsection{Expressing $\mathsf{Z}_d^{K_0}(K)$ and 
$\mathsf{B}_d^{K_0}(K)$ as Sets of Linear Combinations with Certain
Constrains}
The set of cycles and the set of boundaries of $K$ carried
by $K_0$ are
\begin{eqnarray*}
\mathsf{Z}_d^{K_0}(K)&=&\mathsf{Z}_d(K)\cap \mathsf{C}_d(K_0)\ {\rm and}\\
\mathsf{B}_d^{K_0}(K)&=&\mathsf{B}_d(K)\cap \mathsf{C}_d(K_0),
\end{eqnarray*}
respectively. Since $\mathsf{Z}_d(K)$, $\mathsf{B}_d(K)$ and 
$\mathsf{C}_d(K_0)$ are all vector spaces, $\mathsf{Z}_d^{K_0}(K)$
and $\mathsf{B}_d^{K_0}(K)$ are both vector spaces. Furthermore,
since $\mathsf{B}_d(K)$ is a subspace of $\mathsf{Z}_d(K)$,
$\mathsf{B}_d^{K_0}(K)$ is a subspace of $\mathsf{Z}_d^{K_0}(K)$.
It is not hard to show that the subcomplex $K_0$ carries nonbounding cycles of $K$ if and only if 
the dimensions of these two vector spaces are different.

We want to express these two vector spaces as linear combinations
such that we can compute their dimensions using algebraic tools.
We first express the vector spaces, $\mathsf{B}_d(K)$ and 
$\mathsf{Z}_d(K)$ as sets of linear combinations.
Since $\mathsf{B}_d(K)$ is the column space of $\partial_{d+1}$,
a boundary of $K$ can be written as the linear combination of 
column vectors of $\partial_{d+1}$. 
The boundary group can be written as the set of linear combinations
\begin{eqnarray*}
\mathsf{B}_d(K)=\{\partial_{d+1}\gamma \mid \gamma\in \mathbb{Z}_2^{n_{d+1}}\}.
\end{eqnarray*}

The cycle group $\mathsf{Z}_d(K)$ is the union of $\mathsf{B}_d(K)$ and all
the nonbounding cycles of $K$. Suppose we are given a basis for 
$\mathsf{H}_d(K)$, $\{h_1,...,h_{\beta_d}\}$, together with a cycle
for each $h_i$, namely, $z_i\in h_i$. Elements in $h_i$ can be written as
$z_i+\partial_{d+1}\gamma$. Furthermore, elements in $\mathsf{Z}_d(K)$
can be written as linear combinations of $\{b_1,...,b_{n_{d+1}},z_1,...,z_{\beta_d}\}$,
where the $b_j$'s are the column vectors of $\partial_{d+1}$.
We have
\begin{eqnarray*}
\mathsf{Z}_d(K)=\{\hat{Z}_d\gamma \mid \gamma\in \mathbb{Z}_2^{(n_{d+1}+\beta_d)}\},
\end{eqnarray*}
where $\hat{Z}_d=[\partial_{d+1},\hat{H}_d]$ and $\hat{H}_d=[z_1,...,z_{\beta_d}]$.
\begin{remark} 
In our algorithm, the boundary matrix $\partial_{d+1}$ is 
given. We can also precompute
the matrix $\hat{H}_d$ by computing an arbitrary basis of $\mathsf{H}_d(K)$ and
representative cycles of classes in this basis. More details will be provided
in Section \ref{sec:improvedAlg}.
\end{remark}

Since $\mathsf{C}_d(K_0)$ is the set of chain vectors whose $i$-th entry is
zero for any simplex $\sigma_i\notin K_0$, we can write $\mathsf{Z}_d^{K_0}(K)$
and $\mathsf{B}_d^{K_0}(K)$ as elements of $\mathsf{Z}_d(K)$
and $\mathsf{B}_d(K)$ whose $i$-th entries are zero. Consequently,
we can write them as linear combinations with certain constraints,
\begin{eqnarray*}
\mathsf{B}_d^{K_0}(K)&=&\{\partial_{d+1}\gamma \mid \gamma\in \mathbb{Z}_2^{n_{d+1}},
\partial_{d+1}^i\gamma=0 \forall \sigma_i\notin K_0\}\\
\mathsf{Z}_d^{K_0}(K)&=&\{\hat{Z}_d\gamma \mid \gamma\in 
\mathbb{Z}_2^{n_{d+1}+\beta_d},\hat{Z}_d^i\gamma=0 
\forall \sigma_i\notin K_0\}
\end{eqnarray*}
where $\partial_{d+1}^i$ and $\hat{Z}_d^i$ are the $i$-th rows of the matrices 
$\partial_{d+1}$ and $\hat{Z}_d$, respectively.

\subsubsection{Computing Dimensions by Computing Ranks of Sparse Matrices}
With the following theorem, we can compute the dimensions of these
two vector spaces $\mathsf{Z}_d^{K_0}(K)$ and $\mathsf{B}_d^{K_0}(K)$ 
by matrix rank computations.
\begin{theorem}
\label{thm:rank}
For any matrix $A=\inlineM{A_1\\A_2}$, 
$\dim(\{A\gamma\mid A_2\gamma=0\})=\rank(A)-\rank(A_2)$
\end{theorem}
\begin{proof}
For simplicity, denote $\alpha$ as $(\rank(A)-\rank(A_2))$.
There are $\rank(A)$ linearly independent rows in $A$, $\rank(A_2)$ 
linearly independent rows in $A_2$. Therefore, there are 
$\alpha$ rows in $A_1$ that are linearly independent, and
not linearly dependent on rows of $A_2$. 
Choose one such set of rows from $A_1$, 
$A_1'=\inlineM{a_1\\a_2\\ \cdots\\a_{\alpha}}$. 
Since all the rows of $A$ are dependent on rows in $A_1'$ and $A_2$,
for any $\gamma\in \nullspace(A_2)$, $A\gamma$ is determined
by $A_1'\gamma$. 

Proving the theorem is equivalent to showing that $A_1'\gamma$ can be an arbitrary vector
in the vector space $\mathbb{Z}_2^{\alpha}$. It is sufficient to show that for any row of $A_1'$,
$a_i$, the following two statements are both true:
\begin{enumerate}
\item There exist $\gamma_0,\gamma_1\in \nullspace(A_2)$, 
such that $a_i\gamma_0=0$ and $a_i\gamma_1=1$.

\item For any $\gamma\in \nullspace(A_2)$, $a_i\gamma$ does not linearly
depend on the products $a_j\gamma$ for the rest of the rows $a_j$ in $A_1'$.
\end{enumerate}

For the first statement, choose $\gamma_0=0\in \nullspace(A_2)$, which
satisfies $a_i\gamma_0=0$. Now we show that $\gamma_1$ exists
by contradiction. Suppose $a_i\gamma=0$ for all $\gamma\in \nullspace(A_2)$.
This implies that 
\begin{eqnarray*}
&&\nullspace(A_2)\subseteq \nullspace(\inlineM{a_i\\A_2})\\
&\Rightarrow& \rank(\inlineM{a_i\\A_2})\leq \rank(A_2).
\end{eqnarray*}
This contradicts the linear independence of $a_i$ with regard to $A_2$. Therefore,
$a_i\gamma$ can be either $0$ or $1$ for $\gamma\in \nullspace(A_2)$.
In fact, this statement is generally true for any row vector $a$ which is linearly
independent of the rows in $A_2$.

For the second statement, again we prove by contradiction. Suppose 
$a_i\gamma=\sum(a_j\gamma)$ for some rows of $A_1'$, the $a_j$'s.
Define a row vector $a_0=a_i-\sum(a_j)$. We have 
\begin{eqnarray*}
a_0\gamma=(a_i-\sum(a_j))\gamma=0.
\end{eqnarray*} 
Since $a_0$ is linearly independent
of $A_2$, this contradicts to the first statement we have just proved. 
By contradiction, the second statement is true.

In conclusion, for all $\gamma\in \nullspace(A_2)$, $A\gamma$ 
depends on $A_1'\gamma$, whose range space has dimension $\alpha$. 
\end{proof}

It is trivial to see that the order of the rows in these matrices does not 
interfere with the correctness of the theorem. Consequently, the
matrix $A_2$ can be a certain subset of the rows of $A$, not necessarily
the last few rows. Therefore, we can compute the dimensions of 
$\mathsf{B}_d^{K_0}(K)$ and $\mathsf{Z}_d^{K_0}(K)$ as
\begin{eqnarray*}
\dim(\mathsf{B}_d^{K_0}(K))&=&\rank(\partial_{d+1})-\rank(\partial_{d+1}^{K\backslash K_0}),{\rm and}\\
\dim(\mathsf{Z}_d^{K_0}(K))&=&\rank(\hat{Z}_d)-\rank(\hat{Z}_d^{K\backslash K_0}),
\end{eqnarray*}
where $\partial_{d+1}^{K\backslash K_0}$ and $\hat{Z}_d^{K\backslash K_0}$ are the matrices
formed by rows of $\partial_{d+1}$ and $\hat{Z}_d$ whose corresponding
simplices do not belong to $K_0$.

We test whether $K_0$ carries any nonbounding cycle of $K$ by
testing whether these two dimensions are different.
As we know, columns in $\hat{H}_d$ correspond to 
$\beta_d$ nonbounding cycles whose classes form a homology basis.
Therefore, the ranks of $\hat{Z}_d$ and $\partial_{d+1}$ differ by $\beta_d$.
$K_0$ carries nonbounding cycles of $K$ if and only if 
\begin{eqnarray*}
\rank(\hat{Z}_d^{K\backslash K_0})-\rank(\partial_{d+1}^{K\backslash K_0}) \neq \beta_d.
\end{eqnarray*}

\subsubsection{Procedure {\sf Contain-Nonbounding-Cycle($K$,$K_0$,$\hat{H}_d$)}}
With all the facts in hand, we are now ready to state the algorithm
for testing whether a subcomplex carries any nonbounding cycle of $K$.
We use the algorithm of Wiedemann \cite{Wiedemann86} for the rank computation. 
See Algorithm \ref{alg:containNonbounding} for
the pseudocode.
\begin{algorithm}[h!]
 \caption{ {\sf Contain-Nonbounding-Cycle($K$,$K_0$,$\hat{H}_d$)}}
 \label{alg:containNonbounding}

 \begin{algorithmic}[1]
   \GOAL   test whether $K_0$ carries nonbounding cycles of $K$.

    \INPUT $K$: the given simplicial complex.\\
    			 $K_0$: the subcomplex.\\
    			 $\hat{H}_d$: $\beta_d$ linearly independent nonbounding cycles of $K$.

    \OUTPUT Boolean.
		
		\STATE $\hat{Z}_d=[\partial_{d+1},\hat{H}_d]$
		\STATE compute $\partial_{d+1}^{K\backslash K_0}$ and $\hat{Z}_d^{K\backslash K_0}$ by picking
		rows of $\partial_{d+1}$ and $\hat{Z}_d$ whose corresponding simplices do not belong to $K_0$
    \IF{ $\rank(\hat{Z}_d^{K\backslash K_0})-\rank(\partial_{d+1}^{K\backslash K_0})\neq \beta_d$ }
			\STATE return true
		\ELSE
			\STATE return false
		\ENDIF
 \end{algorithmic}
\end{algorithm}

%%%%%%%%%%%%%%%%%%%%%%%%%%%%%%%%%%%%%%%%%%%%%%%%%%%

\subsection{The Improved Algorithm}
\label{sec:improvedAlg}

Next we present the improved version of the procedure {\sf Bmin(K)}.
Theorem \ref{thm:neighborClose} suggests performing a breadth-first search with
a global variable $r_{min}$ and testing whether $B_{p}^{r_{min}-1}$ contains
nonbounding cycles of $K$ for each $p$.
We use the procedure {\sf Contain-Nonbounding-Cycle($K$,$K_0$,$\hat{H}_d$)}
presented in the previous subsection for the testing. See Algorithm \ref{alg:smallestFast}.
\begin{algorithm}[h!]
 \caption{ {\sf Bmin(K)}}
 \label{alg:smallestFast}

 \begin{algorithmic}[1]
 \GOAL  computing $B_{min}(h_{min})$, improved version.

    \INPUT $K$: the given simplicial complex.

    \OUTPUT $p_{min}$,$r_{min}$:the center and radius of $B_{min}(h_{min})$.

    \STATE precompute $\hat{H}_d$
		\STATE compute a breadth-first ordering of $\vertex(K)$, $(p_1,...,p_{n_0})$.
    \STATE apply the persistent homology algorithm on $K$ with filter function $f_{p_1}$
    \STATE $r_{min}$ $=$ the birth time of the first essential homology class
    \STATE $p_{min}$ $=$ $p_1$
    \FOR{$i=2$ to $n_0$}
			\IF{ {\sf Contain-Nonbounding-Cycle($K$,$B_{p_i}^{r_{min}-1}$,$\hat{H}_d$)} }
				\STATE $r_{min}$ $=$ $r_{min}-1$
				\STATE $p_{min}$ $=$ $p_i$
			\ENDIF
		\ENDFOR
 \end{algorithmic}
\end{algorithm}

\paragraph{ Precomputing $\hat{H}_d$.}
The improved algorithm requires the computation of the matrix $\hat{H}_d$, which consists of
$\beta_d$ nonbounding cycles representing elements of a basis of $\mathsf{H}_d(K)$.
For this purpose, any basis is acceptable.
We can precompute $\hat{H}_d$ in a similar way to the 
procedure {\sf Localized-Cycle($p_{min}$,$r_{min}$,$K$)} 
(Algorithm \ref{alg:Localized-Cycle}). More specifically, we perform a column reduction 
on the boundary matrix
$\partial_d$ to compute a basis for the cycle group $\mathsf{Z}_d(K)$. 
We check elements in this basis one by one until we collect $\beta_d$ of them
forming $\hat{H}_d$. For each cycle $z$ in this cycle basis, we check whether
$z$ is linearly independent of the $d$-boundaries and the nonbounding cycles we 
have already chosen, i.e. whether
\begin{eqnarray*}
\rank([z,\partial_{d+1},\hat{H}_d'])\neq \rank([\partial_{d+1},\hat{H}_d']),
\end{eqnarray*}
where $\hat{H}_d'$ consists of cycles we have already chosen for $\hat{H}_d$.
More details are omitted due to the space limitation

\subsection{Complexity}
We analyze the complexity of the improved algorithm.
Denote $n$ and $m$ as the cardinalities
of $K$ and $K'$, respectively. As we know, $m=O(\beta_dn)$.
Similar to the analysis of the non-refined algorithm, the improved algorithm 
{\sf Measure-All($K$)} runs the procedures {\sf Bmin} and {\sf Localized-Cycle}
$\beta_d$ times, with $K'$ as the input. The procedure {\sf Localized-Cycle} 
takes $O(m^3\log^2m)$ time.

The improved procedure {\sf Bmin} precomputes $\hat{H}_d$ once, applies 
the persistent homology algorithm on $K'$ once, and runs the procedure
{\sf Contain-Nonbounding-Cycle} $O(n)$ times. Precomputing $\hat{H}_d$
runs the rank computation $O(m)$ times on matrices with $O(m+\beta_d)=O(m)$ columns
and $O(\beta_dm)$ nonzero entries, and thus takes $O(m^3\log m(\beta_d+\log m))$ time.
The persistent homology algorithm takes $O(m^3)$ time. The 
procedure {\sf Contain-Nonbounding-Cycle} performs rank computations on
matrices with $O(m+\beta_d)=O(m)$ columns
and $O(\beta_dm)$ nonzero entries, and thus takes $O(m^2\log m(\beta_d+\log m))$ time.
Therefore, the procedure {\sf Bmin} takes 
$O(m^3\log m(\beta_d+\log m)+m^3+nm^2\log m(\beta_d+\log m))=O(m^3\log m(\beta_d+\log m))$ time.

Therefore, the whole improved algorithm takes 
$O(\beta_dm^3\log m(\beta_d+\log m))=O(\beta_d^4n^3\log^2 n)$
time.

%%%%%%%%%%%%%%%%%%%%%%%%%%%%%%%%%%%%%%%%%%%%%%%%%%%
%%%%%%%%%%%%%%%%%%%%%%%%%%%%%%%%%%%%%%%%%%%%%%%%%%%

\section{Consistency with Existing Works in Low Dimension}
\label{sec:lowDimResult}

Erickson and Whittlesey \cite{EricksonW05} measured a 1-dimensional
homology class using the length of its shortest cycle. 
They computed the optimal homology basis by finding the set
of nonbounding and linearly independent cycles whose lengths have the
minimal sum. Their algorithm works for $1$-dimensional homology classes
in $2$-manifolds.

We prove in Theorem \ref{thm:optimalInLowDim} that our measure, $S(h)$, is quite close 
to their measure for $1$-dimensional homology classes. For ease of exposition, 
we first prove in Lemma \ref{lem:modifiedAlg} that
by slightly modifying our algorithm of computing the localized cycle,
we can localize the smallest $1$-dimensional homology class, $h_{min}$,
with a representative cycle whose length is no more than $2S(h)+1$.
We start with the modification.

Recall that in the procedure {\sf Localized-Cycle($p_{min}$,$r_{min}$,$K$)},
a localized cycle of $h_{min}$ is computed, given the smallest geodesic ball
carrying $h_{min}$, $B_{min}(h_{min})$, whose center and radius are $p_{min}$
and $r_{min}$, respectively. More specifically, we compute a basis of
the cycles carried by $B_{min}(h_{min})$ by performing a column reduction on
$\partial_d'$, a submatrix of the boundary matrix, $\partial_d$. The submatrix
is constructed by picking columns of $\partial_d$ whose corresponding 
simplices belong to $B_{min}(h_{min})$.

\paragraph{A Modification} When the relevant dimension $d=1$, we modify our 
algorithm as follows. Before performing a column reduction on the submatrix 
$\partial_1'$, we sort its rows and columns in ascending order according to
the function value $f_{p_{min}}$ of their corresponding 1-simplices, 
that is, edges. For edges with the same function value, we sort them 
in ascending order according to the minimal function value of their vertices.
After the sorting, we perform a column reduction on $\partial_1'$ to compute
a basis for the cycles carried by $B_{min}(h_{min})$. The rest is the
same as the original algorithm.

Next, we prove that this modification will produce a localized cycle 
of $h_{min}$ whose length is no greater than $2S(h_{min})+1$.

\begin{lemma}
\label{lem:modifiedAlg}
The modified algorithm localizes the smallest $1$-dimensional 
homology class, $h_{min}$, with a 1-cycle with no more than $2S(h_{min})+1$ edges.
\end{lemma}
\begin{proof}
For simplicity, 
we prove the case when $K$ has only one connected component.
The general case follows simply.

Because of the properties of the geodesic distance, we observe the
following two facts.
\begin{enumerate}
\item For any edge, the function values of its vertices differ in no
more than $1$.

\item For each vertex, $q\neq p_{min}$, there exists
at least one edge with vertices $q$ and $q'$, such that
\begin{eqnarray*}
f_{p_{min}}(q')=f_{p_{min}}(q)-1.
\end{eqnarray*} 
By {\it lower edges}, we denote edges whose two vertices have different function 
values.
\end{enumerate}

These facts imply that in the modified algorithm, a column is reduced
to a nonzero column only if its corresponding edge is a lower edge. 
To see this, notice that in the simplex-ordering corresponding to the
sorted $\partial_d'$, for any vertex
$q\neq p_{min}$, among all the edges adjacent to it, lower edges
must appear first. During the reduction, 
$q$ must be paired with one of its lower edges. Since $p_{min}$
corresponds to the $0$-dimensional essential homology class,
it is not paired by any edge. Therefore, any edge paired with
a vertex is a lower edge. Any column which is reduced to a
nonzero column corresponds to a lower edge.

The localized cycle we compute, $z_{min}$, is one of the columns
of $V$, corresponding to zero columns in $R$, where 
$R=\partial_1'V$. Let it be the $i$-th column, corresponding to 
$\sigma_i$. It is straightforward
to see that only columns corresponding to lower edges are used to 
reduce column $i$ of $\partial_1'$. Consequently, in the computed
localized cycle, any edge beside $\sigma_i$ is a lower edge, and 
thus has two vertices whose function values differ in one. Since 
edge $i$ has the function value $S(h_{min})$, $z_{min}$ has no
more than $2S(h_{min})+1$ edges.
\end{proof}

For example, in Figure \ref{fig:smallerCircle}, $B_{min}(h_{min})$ is
centered at $p_1$ with radius two. Using the modified algorithm,
edge $p_3p_4$ corresponds to the nonbounding cycle. Its column is reduced
using edges $p_1p_2$, $p_2p_3$, $p_4p_5$ and $p_1p_5$, which are all
lower edges. The computed localized cycle has length $5=2S(h_{min})+1$.
\begin{figure}[hbtp]
    \centering
    \includegraphics[width=0.4\textwidth]{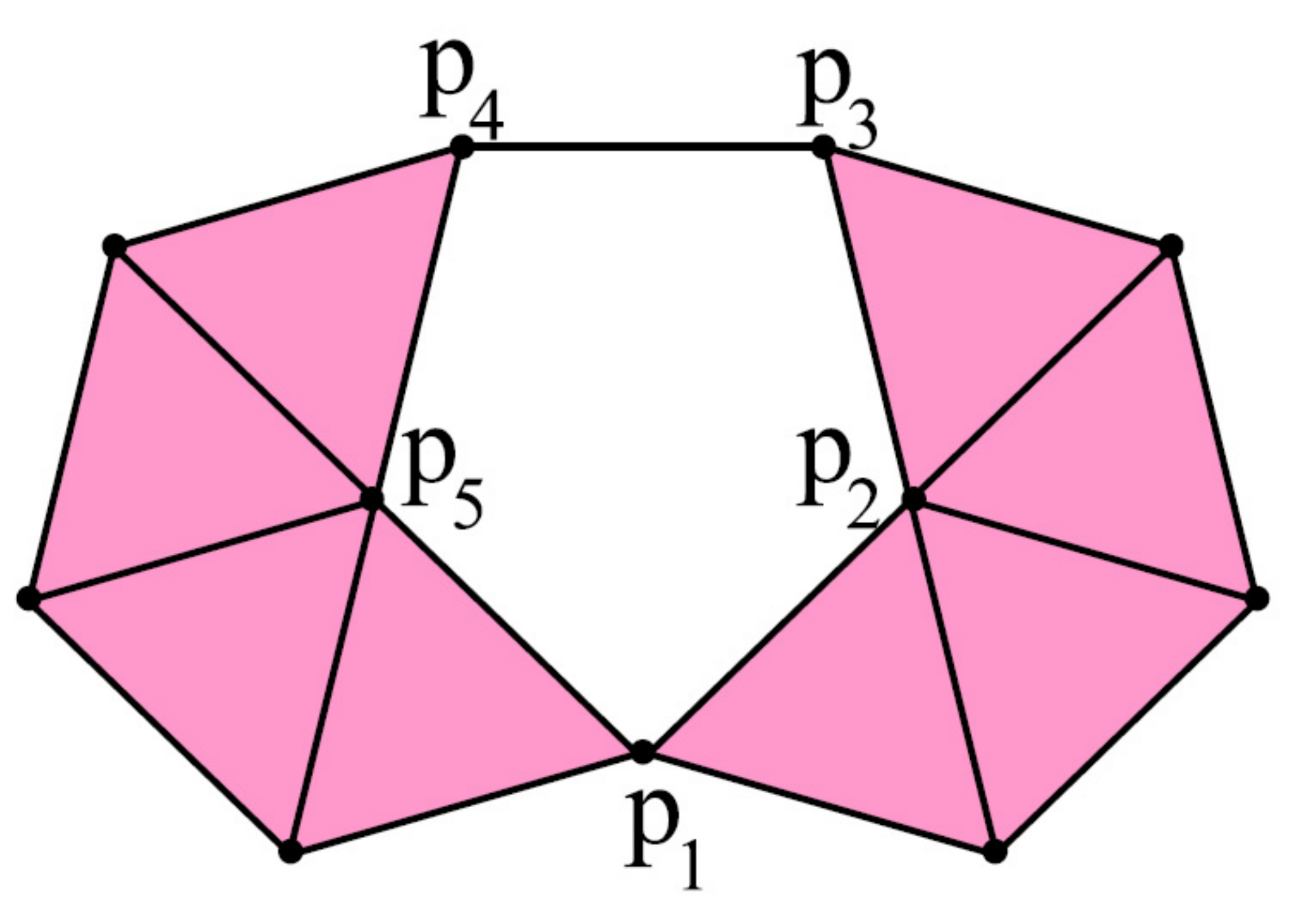}
    \caption{ Edge $p_3p_4$ corresponds to the localized cycle whose length
    is $2S(h_{min})+1$.}
    \label{fig:smallerCircle}
\end{figure}

Based on this Lemma, we prove that our result is close to the result of \cite{EricksonW05},
in which size of a $1$-dimensional homology class is the length or its shortest
representative cycle, namely,
\begin{eqnarray*}
S_E(h)=\min_{z\in h} \length(z), h\in \mathsf{H}_1(K).
\end{eqnarray*}
\begin{theorem}
For a $1$-dimensional homology class $h$,
\begin{eqnarray*}
2S(h)\leq S_E(h)\leq 2S(h)+1.
\end{eqnarray*}
\label{thm:optimalInLowDim}
\end{theorem}
\begin{proof}
Lemma \ref{lem:modifiedAlg} shows that there exists a representative cycle of $h$ with
no more than $2S(h)+1$ edges. Therefore, the shortest representative cycle of $h$ has
no more than $2S(h)+1$ edges. We have
\begin{eqnarray*}
S_E(h)\leq 2S(h)+1.
\end{eqnarray*}

Next, we show that 
\begin{equation}
2S(h)\leq S_E(h).
\label{eqn:size'BoundsSize}
\end{equation}
Pick the shortest representative cycle $z_0$ with length $S_E(h)$. Choose
any vertex $p\in z_0$ as the center to build a smallest geodesic ball
carrying $z_0$. The radius of this ball is $S_E(h)/2$ when 
$S_E(h)$ is even, and $(S_E(h)-1)/2$ when $S_E(h)$ is odd.
Since $S(h)$ is no greater than this radius, 
Equation (\ref{eqn:size'BoundsSize}) is proved.
\end{proof}

This theorem shows that our measure tightly bounds the one by Erickson 
and Whittlesey. Furthermore, we know the localized cycles computed are 
almost the shortest ones.
\begin{corollary}
The localized cycle of $h$ computed by the modified algorithm has at most 
one more edge than the shortest representative cycle of $h$.
\end{corollary}

\begin{remark}
In fact, the algorithm can be further modified to generate exactly
the same result as the one by Erickson and Whittlesey. We omit
this because it involves more technical details and does not
provide any new insights.
\end{remark}
\begin{remark}
Our modified algorithm can compute the shortest representative cycle for 1-dimensional homology
classes no matter what dimension $K$ is, whereas most of the existing works in low
dimension require $K$ to be dimension two.
\end{remark}

%%%%%%%%%%%%%%%%%%%%%%%%%%%%%%%%%%%%%%%%%%%%%%%%%%
%%%%%%%%%%%%%%%%%%%%%%%%%%%%%%%%%%%%%%%%%%%%%%%%%%

\section{Conclusion}
In this paper, we have defined a size measure of homology classes, found cycles
localizing these classes, as well as computed an optimal homology basis for the homology group.
An $O(\beta^4 n^4)$ brute force algorithm has been presented, which measures and localizes
the optimal homology basis by applying the persistent
homology algorithm on the simplicial complex $\beta n$ times. Aided by
Theorem \ref{thm:neighborClose} and \ref{thm:rank}, we have improved the algorithm to
$O(\beta^4 n^3\log^2 n)$. Finally, we have shown that our result is similar
to the existing optimal result in low dimensions.

\paragraph{ Future directions.}
We intend to extend our work in two directions.
\begin{enumerate}
\item In this paper, a localized cycle $z_0\in h$ satisfies the condition
\begin{eqnarray*}
\rad(z_0)=\min_{p\in K}\max_{q\in \vertex(z_0)} dist(p,q)=\min_{z\in h} \rad(z).
\end{eqnarray*}
Can we localize $h$ with a representative cycle using other size measures? Examples of such 
measures are:
\begin{eqnarray*}
\card(z_0)&=&\min_{z\in h} \card(z),\\
\diam(z_0)&=&\max_{p,q\in \vertex(z_0)} dist(p,q)=\min_{z\in h} \diam(z),{\rm and}\\
\radZ(z_0)&=&\min_{p\in \vertex(z_0)}\max_{q\in \vertex(z_0)} dist_{z_0}(p,q)=\min_{z\in h} \radZ(z),
\end{eqnarray*}
where $\card(z)$ is number of simplices in the cycle $z$ and
$dist_{z_0}(q,p)$ is the geodesic distance between $p$ and $q$
within the representative cycle $z_0$.
We conjecture computing $z_0$ satisfying the first two constraints are NP-complete.

\item Can we extend the results if we replace the discrete geodesic distance with
continuous metric defined on the underlying space of the simplicial complex?
\end{enumerate}

\section*{Acknowledgment}
We would like to thank Herbert Edelsbrunner for helpful discussion on computing 
representative cycles of homology classes in the persistent homology algorithm.

\bibliographystyle{abbrv}
\bibliography{topologyBib}
\end{document}